\documentclass[11pt,letterpaper]{article}
\usepackage{neutralpreprint}
\usepackage[utf8]{inputenc}
\usepackage[T1]{fontenc}
\usepackage{hyperref,url,booktabs}
\usepackage{amsmath,amssymb,amsthm,mathtools,bm}
\usepackage{graphicx,xcolor,microtype}
\usepackage{algorithm,algorithmic,placeins}
\hypersetup{hypertexnames=false,
  pdftitle={Adaptive mixture variational inference for spike-and-slab regression}, pdfauthor={Hanqing Li, Yaroslav Golub, and Xuewen Lu}, pdfsubject={Variational inference and uncertainty quantification for sparse Bayesian regression}, pdfkeywords={Variational inference, spike-and-slab priors, Bayesian variable selection, mixture approximations, high-dimensional regression, Bernstein-von Mises theorem}}
\newcommand{\R}{\mathbb R}
\newcommand{\E}{\mathbb E}
\newcommand{\KL}{D_{\mathrm{KL}}}
\newcommand{\TV}{d_{\mathrm{TV}}}

\newcommand{\one}{\mathbf 1}

\theoremstyle{preprintplain}
\newtheorem{theorem}{Theorem}
\newtheorem{proposition}{Proposition}
\newtheorem{corollary}{Corollary}

\theoremstyle{preprintdefinition}

\newtheorem{assumption}{Assumption}

\title{Adaptive mixture variational inference for\\spike-and-slab regression}
\author{\href{https://orcid.org/0009-0001-0493-237X}{Hanqing Li}\textsuperscript{1}\qquad
  \href{https://orcid.org/0009-0007-6398-4421}{Yaroslav Golub}\textsuperscript{2}\qquad
  \href{https://orcid.org/0000-0002-0905-2697}{Xuewen Lu}\textsuperscript{1}\\[4pt]
  {\small\textsuperscript{1}Department of Mathematics and Statistics, University of Calgary}\\[2pt] {\small\textsuperscript{2}Department of Electrical and Software Engineering, University of Calgary}\\[2pt] {\small\href{mailto:hanqing.li@ucalgary.ca}{hanqing.li@ucalgary.ca}\qquad
  \href{mailto:yaroslav.golub@ucalgary.ca}{yaroslav.golub@ucalgary.ca}\qquad
  \href{mailto:xlu@ucalgary.ca}{xlu@ucalgary.ca}}}

\date{}
\begin{document}
\maketitle

\begin{abstract}
Correlated predictors can support competing sparse explanations with similar predictions, making joint uncertainty about variable inclusion difficult to capture with mean-field approximations. We develop an adaptive fitting procedure for mixtures of product distributions in Gaussian regression with a point-mass spike-and-slab prior. It minimizes reverse Kullback--Leibler divergence directly on inclusion indicators and active coefficients, jointly refining component parameters and weights as the mixture grows. This avoids an additional divergence penalty on unused latent coefficients under independent augmentation. Our analysis relates approximation accuracy to mixture size, support coverage and dependence within supports, and establishes contraction, selection consistency and a Bernstein--von Mises approximation under explicit conditions on the prior, posterior concentration and variational error. On all 250 simulated datasets with exact posterior references, mixtures reduce errors in inclusion probabilities, grouped support probabilities and coefficient covariance relative to multistart mean field. Comparisons at fixed mixture size and common initialization favor direct joint refinement over augmented or restricted refinement in posterior divergence. Complete stagewise fitting can nevertheless be more accurate near collinearity. The results support direct joint refinement for posterior approximation while showing that local gains do not ensure superiority of the full adaptive search.

\end{abstract}

\begingroup\small
\par
\endgroup

\section{Introduction}
In sparse regression, accurate prediction need not resolve which predictors explain the response. Two highly correlated predictors may serve as substitutes: the posterior can favor including either one while assigning little probability to including both. A useful approximation should preserve this uncertainty about model choice alongside uncertainty in the coefficients. We study Gaussian regression with a point-mass spike-and-slab prior, which expresses both forms of uncertainty \citep{mitchell1988bayesian,george1993variable}.

Mean-field variational inference (MFVI) minimizes reverse Kullback--Leibler (KL) divergence over coordinate products \citep{blei2017variational}, enabling efficient spike-and-slab fitting and sparse estimation guarantees \citep{carbonetto2012scalable,spence2020flexible,ray2022vb}. Under independence, preserving the marginal inclusion probabilities of two competing predictors also assigns probability to selecting both or neither. Thus accurate estimation can coexist with distorted joint uncertainty \citep{giordano2018covariances,margossian2025uncertainty}. Other approaches represent dependence through single-effect decompositions or entropic regularization \citep{wang2020susie,wu2026xi}.

Mixtures offer a natural way to represent competing explanations: components remain simple, while their average captures dependence \citep{bishop1997mixtures,miller2017boosting,locatello2018boosting}. Gaussian-mixture and variable-selection methods already exploit this flexibility \citep{arenz2020trust,arenz2023unified,rockova2016determinantal,henclova2026gemss}. In particular, the augmented boosting baseline of \citet{spence2020flexible} induces our exact-zero mixture family when all parameters and weights vary freely. Our focus is therefore how the objective and component updates affect approximation within this family.

Two considerations guide our approach. An augmented approximation that makes inclusion indicators and latent coefficients independent can incur a divergence cost for excluded predictors, whose latent coefficients do not affect the likelihood. An early component may cover several competing explanations; adding another can allow the earlier component to specialize, provided its parameters and weight remain adjustable. We therefore extend the direct treatment of inclusion indicators and active coefficients in Spence's product method \citep{spence2020flexible} to mixtures, jointly refining all component parameters and weights as the mixture grows.

We develop an adaptive algorithm implementing these principles, with separate numerical validation of proposed improvements. Our analysis relates attainable accuracy to mixture size, support coverage and within-support dependence, and gives variational-error conditions for contraction, selection consistency and a Gaussian limit under explicit model and posterior assumptions. On exact-reference datasets, mixtures reduce errors in support probabilities and coefficient covariance relative to mean field. Controlled comparisons from a common initialization favor direct joint refinement in reverse KL, while complete stagewise fitting can perform better near collinearity. These findings distinguish local refinement gains from the performance of the full search procedure.

\section{Model and mixtures}
\label{sec:formulation}

\subsection{Spike-and-slab regression}
\label{sec:model}
Let $y\in\R^n$ be the observed response and $X\in\R^{n\times p}$ the fixed design matrix, with column $X_j$. We use Gaussian regression with an independent point-mass spike-and-slab prior,
\begin{equation}
 y\mid\beta\sim N_n(X\beta,\sigma^2I_n),\quad \beta_j\mid\gamma_j\sim(1-\gamma_j)\delta_0+\gamma_jN(0,\tau^2),\quad \gamma_j\sim\operatorname{Bernoulli}(\omega).
 \label{eq:model}
\end{equation}
Here $j=1,\ldots,p$, $I_n$ is the identity matrix, and $\delta_0$ is unit mass at zero. The noise variance $\sigma^2>0$, slab variance $\tau^2>0$, and inclusion probability $0<\omega<1$ are specified for each dataset. Define the indicator--coefficient pair $z_j=(\gamma_j,\beta_j)$, and $z=(z_1,\ldots,z_p)$. The state space of coordinate $z_j$ is $\mathcal S_j=\{(0,0)\}\sqcup(\{1\}\times\R)$, where $\sqcup$ denotes disjoint union. To treat the atom and continuous slab together, let $\nu_j$ assign unit mass to $(0,0)$ and Lebesgue measure to the active copy of $\R$.

Write $\varphi(b;m,v)$ for the $N(m,v)$ density. The prior $P_j$ has $\nu_j$-density $1-\omega$ when inactive and $\omega\varphi(\beta_j;0,\tau^2)$ when active. With $P=\bigotimes_jP_j$ and $\nu=\bigotimes_j\nu_j$, the posterior is
\begin{equation}
 \frac{\mathrm{d}\Pi}{\mathrm{d}P}(z)=\frac{e^{-\ell(z)}}{\mathcal Z},\qquad
 \ell(z)=\frac{\|y-X\beta\|^2}{2\sigma^2},\qquad
 \mathcal Z=\int e^{-\ell(u)}P(\mathrm{d}u).
 \label{eq:posterior}
\end{equation}
For probability measures $Q$ and $R$, define $\KL(Q\|R)=\int\log(\mathrm{d}Q/\mathrm{d}R)\,\mathrm{d}Q$ when $Q\ll R$, and infinity otherwise. Total variation (TV) distance is $\TV(Q,R)=\sup_A|Q(A)-R(A)|$. Norms without a subscript are Euclidean. The variational objective is
\begin{equation}
 \mathcal L(Q)=\E_Q\ell+\KL(Q\|P)=\KL(Q\|\Pi)-\log\mathcal Z.
 \label{eq:objective}
\end{equation}
Thus minimizing $\mathcal L$ approximates the posterior without evaluating $\mathcal Z$. For two candidate distributions, an objective difference is exactly their difference in reverse KL. Throughout, each indicator and its coefficient remain one joint coordinate.

\subsection{Mixture family}
\label{sec:mixture-family}
We use uppercase $Q$ for variational probability measures and lowercase $q$ for their densities. For component $k=1,\ldots,K$ and coordinate $j$, let $Q_{kj}$ have inclusion probability $\alpha_{kj}\in[0,1]$, active mean $\mu_{kj}\in\R$, and active variance $v_{kj}>0$. Its coordinate density $q_{kj}=\mathrm{d}Q_{kj}/\mathrm{d}\nu_j$ is
\begin{equation*}
 q_{kj}(z_j)=\begin{cases}1-\alpha_{kj},&z_j=(0,0),\\ \alpha_{kj}\varphi(\beta_j;\mu_{kj},v_{kj}),&z_j=(1,\beta_j).\end{cases}
\end{equation*}
Define the component law $Q_k=\bigotimes_jQ_{kj}$, whose joint $\nu$-density is $q_k(z)=\prod_{j=1}^p q_{kj}(z_j)$, and let
\begin{equation}
 \mathcal Q_K=\left\{Q_\theta=\sum_{k=1}^K w_kQ_k:\ w_k\geq0,\ \sum_{k=1}^K w_k=1\right\}.
 \label{eq:mixture-family}
\end{equation}
The mixture law $Q_\theta$ has $\nu$-density $q_\theta(z)=\sum_{k=1}^K w_kq_k(z)$. The parameter $\theta=(w,\alpha,\mu,v)$ collects the weights and component parameters. A single component is Bernoulli--Gaussian mean field. Allowing zero weights or duplicate components makes the families nested in $K$. Boundary inclusion probabilities are also allowed. Each component can assign positive mass to many regression supports.

Let $C$ be a component label with $\Pr(C=k)=w_k$, and let $Z\mid C=k\sim Q_k$, with realization $z$. The pairs $Z_j$ are independent conditional on $C$ but can be dependent after averaging over $C$. For example, two equally weighted components with inclusion vectors $(1,0)$ and $(0,1)$ represent a law assigning probability $1/2$ to each of these two supports. Both marginal inclusion probabilities are $1/2$. A product law with those marginals instead assigns probability $1/4$ to every support, including neither or both variables.

Write $m_k=(m_{k1},\ldots,m_{kp})^\top$ and $d_k=(d_{k1},\ldots,d_{kp})^\top$ for the component mean and variance vectors, where $m_{kj}=\alpha_{kj}\mu_{kj}$ and $d_{kj}=\alpha_{kj}(v_{kj}+\mu_{kj}^2)-m_{kj}^2$, and let $\bar m=\sum_kw_km_k$. The law of total covariance yields
\begin{equation}
 \operatorname{Cov}_{Q_\theta}(\beta)=\sum_kw_k\operatorname{diag}(d_k)+\sum_kw_k(m_k-\bar m)(m_k-\bar m)^\top.
 \label{eq:mixture-covariance}
\end{equation}
The second term is the between-component covariance and has rank at most $K-1$. Section~\ref{sec:asymptotic} relates this structure to approximation accuracy. The posterior inclusion probabilities (PIPs) are $Q_\theta(\gamma_j=1)=\sum_kw_k\alpha_{kj}$, the posterior means are $\E_{Q_\theta}\beta_j=\sum_kw_k\alpha_{kj}\mu_{kj}$, and joint inclusions satisfy $Q_\theta(\gamma_i=\gamma_j=1)=\sum_kw_k\alpha_{ki}\alpha_{kj}$ for $i\ne j$. Marginal coefficient distributions combine a zero atom and a Gaussian mixture. Their distribution functions are explicit, and quantiles follow by numerical inversion with the jump at zero handled separately. More general joint events can be assessed by direct draws from the fitted mixture.

\subsection{Direct and augmented objectives}
\label{sec:direct-augmentation}
The augmented boosting baseline in \citet{spence2020flexible} uses mixtures of independent Bernoulli and Gaussian variables. With freely varying weights and component parameters, it induces $\mathcal Q_K$. This baseline differs from that paper's principal product method, which already avoids auxiliary independence. We compare its augmented objective with our direct objective on the same induced family.

Write $\beta=\gamma\odot\beta^+$, where $\odot$ denotes coordinatewise multiplication, with independent priors $\beta_j^+\sim N(0,\tau^2)$ and the indicator prior and likelihood in \eqref{eq:model}. Let $\Pi^+$ and $Q^+$ be the augmented posterior and variational law, and $Q$ the induced law of $z=(\gamma,\beta)$. Under $\Pi^+$, conditioning on $z$ leaves $\beta_{S^c}^+$ with prior law $\mathcal N_{S^c}=\bigotimes_{j\notin S}N(0,\tau^2)$, where $S=\{j:\gamma_j=1\}$. For finite divergences, the relative-entropy chain rule gives
\[
 \KL(Q^+\|\Pi^+)=\KL(Q\|\Pi)
 +\E_Q\KL\{Q^+(\beta_{S^c}^+\mid z)\|\mathcal N_{S^c}\}.
\]
The extra term is nonnegative and vanishes exactly when the conditional discarded-variable law equals $\mathcal N_{S^c}$, $Q$-almost surely. For a product law with independent $\gamma_j\sim\operatorname{Bernoulli}(\alpha_j)$ and $\beta_j^+\sim N(\mu_j,v_j)$, it is
\[
 \KL(Q^+\|\Pi^+)-\KL(Q\|\Pi)
 =\sum_j(1-\alpha_j)\KL\{N(\mu_j,v_j)\|N(0,\tau^2)\}.
\]
Thus auxiliary independence penalizes departures of the latent Gaussian from its prior even when a predictor is inactive. Direct optimization removes this cost while retaining the induced Bernoulli--Gaussian mixture family. For a mixture, the discarded-variable conditional law is itself a mixture, so the extra cost is not generally the weighted sum of component costs. Appendix~\ref{app:augmentation} gives the derivation and a completion of any direct law with zero conditional cost. The identity motivates direct fitting but does not establish numerical superiority. Section~\ref{sec:numerical} compares direct and augmented objectives from a common candidate at fixed $K$, alongside comparisons of complete fitting procedures.

\section{Computation}
\label{sec:computation}
Starting from mean field, the procedure adds components and jointly refines their parameters and weights. Separate integration draws assess proposed improvements. Appendix~\ref{app:reference-adaptive} specifies the numerical settings used in Algorithm~\ref{alg:adaptive-mixture}.

\subsection{Objective evaluation}
\label{sec:mixture-objective}
Let $\mathcal L_k=\E_{Q_k}\ell+\KL(Q_k\|P)$, and recall the component and mixture densities $q_k$ and $q_\theta$. Define the mutual information between the component label and the joint vector $Z$ by
\begin{equation}
 \mathcal J(\theta):=\sum_{k=1}^{K}w_k\E_{Q_k}\log\frac{q_k(Z;\theta)}{q_\theta(Z)}.
 \label{eq:mixture-joint-information}
\end{equation}
Zero-weight summands are omitted, and density-weighted logarithmic terms are zero wherever their density factor vanishes. The standard mixture entropy identity evaluates \eqref{eq:objective} as follows.
\begin{proposition}[Objective identity]
\label{prop:mixture-information}
For finite component objectives, with $H(w)=-\sum_kw_k\log w_k$ and the convention $0\log0=0$, we have
\begin{equation*}
 \mathcal L(Q_\theta)=\sum_kw_k\mathcal L_k-\mathcal J(\theta),\qquad 0\leq \mathcal J(\theta)\leq H(w)\leq\log K.
\end{equation*}
\end{proposition}
Each $\mathcal L_k$ is analytic, and its residual calculation uses $X$ without a dense $p\times p$ Gram matrix (Appendix~\ref{app:mixture}). At $K=1$, mutual information vanishes. For larger $K$, only the overlap term requires numerical integration. We use scrambled quasi-Monte Carlo points \citep{owen1995randomly,liu2021quasi}, with independent scrambles assessing integration variability, including that from the Bernoulli transformation.

\subsection{Joint refinement}
\label{sec:mixture-algorithm}
Joint refinement holds the component count $K$ fixed while updating all weights, inclusion probabilities, active means and variances together, including those of existing components. Thus all coordinates of $\theta=(w,\alpha,\mu,v)$ may change. The local counter $t$ starts at zero from a proposed initialization $\theta^{(0)}$ and is distinct from the expansion counter $k$ in Algorithm~\ref{alg:adaptive-mixture}.

\begin{algorithm}[tbp]
\caption{Adaptive mixture variational inference for spike-and-slab regression}
\label{alg:adaptive-mixture}
\begin{algorithmic}[1]
\REQUIRE $(X,y,\sigma^2,\omega,\tau^2)$, component cap $K_{\max}$, search budgets, seed.
\ENSURE Fitted mixture $Q$ with $K$ components, stopping record and independent assessment.
\STATE $Q^{(1)}\gets Q_{\mathrm{MF}}$, the best of ten MFVI starts; $k\gets1$.
\WHILE{$k<K_{\max}$ and the search budget permits another stage}
 \STATE $\mathrm{accepted}\gets\mathrm{false}$.
 \FOR{\textit{proposal} in (small split, large split, residual product), in order}
  \STATE \textbf{if} the search budget is exhausted \textbf{then break}
  \STATE $\widetilde Q^{(k+1)}\gets\text{\textsc{Propose}}(Q^{(k)},\textit{proposal})$.
  \STATE $\widetilde Q^{(k+1)}\gets\text{\textsc{JointRefine}}(\widetilde Q^{(k+1)})$.
  \IF{$\text{\textsc{Validate}}(\widetilde Q^{(k+1)},Q^{(k)};\mathrm{expansion})$}
   \STATE $Q^{(k+1)}\gets\widetilde Q^{(k+1)}$; $k\gets k+1$; $\mathrm{accepted}\gets\mathrm{true}$.
   \STATE \textbf{break}
  \ENDIF
 \ENDFOR
 \STATE \textbf{if} $\mathrm{accepted}=\mathrm{false}$ \textbf{then break}
\ENDWHILE
\STATE Record whether search stopped at the component cap, budget or failed proposals.
\STATE $K\gets k$.
\IF{$k>1$ and $\text{\textsc{Validate}}(Q^{(k)},Q^{(1)};\mathrm{fallback})=\mathrm{false}$}
 \STATE $K\gets1$; record fallback.
\ENDIF
\STATE $Q\gets Q^{(K)}$.
\STATE Independently assess $Q$ using a separate integration batch.
\RETURN $Q$, stopping record and assessment.
\end{algorithmic}
\end{algorithm}

\paragraph{Fixed reference and numerical objective.}
Write the reference mixture as $R=\sum_{h=1}^{K}\widetilde w_hR_h$. If $r_h$ denotes the density of $R_h$, its density is $r(z)=\sum_{h=1}^{K}\widetilde w_hr_h(z)$. At local iteration $t$, choose $R^{(t)}=Q_{\theta^{(t)}}$, with weights $\widetilde w_h=w_h^{(t)}$ and component distributions $R_h=Q_h(\theta^{(t)})$. For the calculations below, write $R=R^{(t)}$. Generate $N$ scrambled quasi-Monte Carlo points $z_{h1},\ldots,z_{hN}$ from each $R_h$. Keep this reference and these points fixed while optimizing the trial parameters $\theta$. The saved weights $\widetilde w_h$ remain fixed while the trial weights $w_k$ may change.

The change of measure requires $Q_\theta\ll R$. This holds during numerical fitting because positive weights and variances and $0<\alpha_{kj}<1$ give $r(z)>0$ throughout the mixed state space. For the trial density $q_\theta(z)=\sum_{k=1}^{K}w_kq_k(z;\theta)$, importance sampling gives
\[
 \mathcal J(\theta)=\E_R\!\left[\sum_{k=1}^{K}
       \frac{w_kq_k(Z;\theta)}{r(Z)}
       \log\frac{q_k(Z;\theta)}{q_\theta(Z)}\right].
\]
Each reference point is evaluated under every trial component. The density ratios account for changes in the trial distribution while the sampling distribution remains fixed. Since $\E_R f=\sum_{h=1}^{K}\widetilde w_h\E_{R_h}f$, replacing each component expectation by its sample average yields
\begin{equation}
 \widehat{\mathcal J}(\theta)=\frac{1}{N}\sum_{h=1}^{K}\widetilde w_h
 \sum_{b=1}^{N}\sum_{k=1}^{K}\frac{w_kq_k(z_{hb};\theta)}{r(z_{hb})}
 \log\frac{q_k(z_{hb};\theta)}{q_\theta(z_{hb})}.
 \label{eq:mixture-importance}
\end{equation}
Here $h$ indexes reference components, $b$ indexes points within them, and $k$ indexes trial components. Since each reference component supplies $N$ points, its average receives weight $\widetilde w_h$. The sum over $k$ evaluates the component labels analytically. The numerical objective is $\widehat{\mathcal L}(\theta)=\sum_{k=1}^{K}w_k\mathcal L_k-\widehat{\mathcal J}(\theta)$.

\paragraph{Joint optimization and validation.}
With the reference fixed, bounded L-BFGS-B \citep{byrd1995limited} updates all weight logits, inclusion logits, active means and log variances. Differentiation includes the trial density ratios and weights, so inclusion probabilities can change without differentiating sampled Bernoulli thresholds. Appendix~\ref{app:mixture-gradients} gives the finite-objective derivatives.

A proposed update must pass overlap and importance-weight checks before local validation compares it with $Q_{\theta^{(t)}}$ using integration batches separate from optimization. For $M\geq2$ independent scrambles, let $\overline\Delta$ be the mean candidate-minus-comparator objective difference and $s_\Delta$ its estimated standard error. Accept when
\begin{equation}
 \overline\Delta<-\max\{10^{-4},3s_\Delta\},
 \label{eq:empirical-acceptance}
\end{equation}
with information estimates consistent with $0\leq\mathcal J(\theta)\leq H(w)$. Retain an accepted update as $\theta^{(t+1)}$. Otherwise set $\theta^{(t+1)}=\theta^{(t)}$. The next local iteration refreshes the reference and points using the retained parameters. Algorithm~\ref{alg:joint-refinement} specifies this loop. The empirical rule assesses integration variability; sufficient error bounds for exact descent appear in Appendix~\ref{app:mixture-acceptance}.

\paragraph{Adaptive component expansion.}
Algorithm~\ref{alg:adaptive-mixture} starts from $Q^{(1)}=Q_{\mathrm{MF}}$. At stage $k$, it tries three $(k+1)$-component candidates from $Q^{(k)}$ in order. Small and large splits replace the largest-weight component by two copies with half its weight and opposite mean displacements at scales $0.5$ and $1$. The residual product appends a residual-guided product with weight $0.1$ and rescales existing weights by $0.9$. Appendix~\ref{app:reference-adaptive} gives construction details.

Each candidate $\widetilde Q^{(k+1)}$ undergoes joint refinement at fixed size, resetting $t$ to zero. Expansion validation then applies \eqref{eq:empirical-acceptance} and the information bounds against the incumbent $Q^{(k)}$ using fresh integration draws. The first accepted candidate becomes $Q^{(k+1)}$ and ends the stage. Here $K$ is the returned component count. Final assessment is separate from fitting and fallback validation.

\paragraph{Restricted refinement strategies.}
Frozen refinement fixes existing component distributions after proposal initialization, updating all weights and the new component. Stagewise fitting optimizes $Q=(1-\zeta)Q^{(k)}+\zeta G$ over a new product $G$ and its weight $\zeta$, holding the incumbent $Q^{(k)}$ fixed, including its relative component weights \citep{miller2017boosting}. Joint refinement updates existing component parameters as well. These restrictions distinguish the fitting strategies compared in Section~\ref{sec:numerical}.

\FloatBarrier

\section{Experiments}
\label{sec:numerical}

We generate 550 independent datasets with $n=80$ and five signals of magnitude $0.7$. Each group contains three predictors with pairwise population correlation $\rho$ and one randomly chosen signal. Groups are mutually independent; all other predictors are independent of one another and of the groups. The one- and two-group designs have four and three additional signals, respectively, with randomly permuted columns. We use 50 datasets per $(p,\rho)$ cell: $p=10,20,30,100$ and $\rho=0.7,0.9$ for one group, and $p=10$ and $\rho=0.7,0.9,0.99$ for two groups. Fits use $\sigma=\tau=1$, $\omega=5/p$, ten MFVI starts and at most ten components.

At $p=10$, support enumeration gives exact references for reverse KL, PIP, grouped-support TV and covariance errors. Larger dimensions use eligible Markov chain Monte Carlo~(MCMC) references for PIP and support errors. Prediction mean squared error (MSE) uses 1,000 independent test rows. Paired intervals are descriptive, unadjusted 95\% $t$ intervals across datasets. Appendix~\ref{app:numerical} gives designs, metrics and reference checks. Code and saved simulation results are publicly available.\footnote{\url{https://github.com/lihanqing1997/adaptive-mixture-spike-slab}}

\subsection{Posterior approximation}
All 550 adaptive mixtures improve the independently evaluated objective over MFVI; the objective difference equals the reverse-KL difference within a dataset. PIP, grouped-support and covariance errors also decrease on every exact-reference dataset. Figure~\ref{fig:joint-errors} shows these dependence gains on the 100 one-group datasets at $p=10$: support TV covers the three grouped predictors, while covariance error covers all coefficients. At $\rho=0.7,0.9$, mean reverse KL falls from $0.377$ to $0.090$ and from $0.692$ to $0.338$, respectively (Table~\ref{tab:formal-posterior}).

\begin{figure}[htbp]
\centering
\includegraphics[width=.82\linewidth]{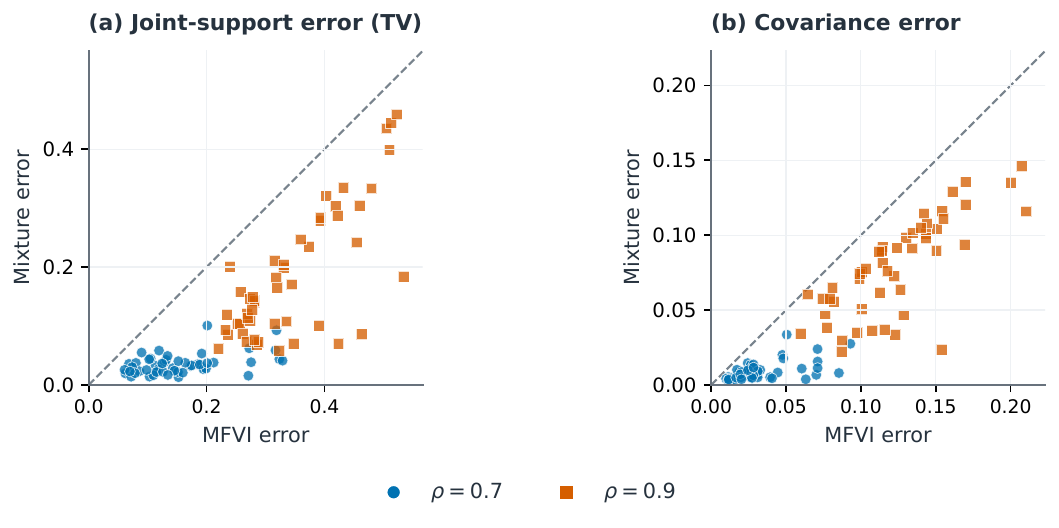}
\caption{One-group support and covariance errors: mixture versus MFVI.}
\label{fig:joint-errors}
\end{figure}

At $p=20,30$, mean PIP errors fall by $24.7$--$41.0\%$ and grouped-support TV by $22.8$--$51.0\%$, with paired intervals below zero. At $p=100$, both error intervals include zero at both correlations (Table~\ref{tab:formal-paired}). Two-group mean KL falls by $67.5\%$, $42.8\%$ and $20.9\%$ relative to MFVI as $\rho$ increases; support TV falls by $62.1\%$, $30.1\%$ and $12.7\%$, but remains $0.646$ at $\rho=0.99$. These approximation gains do not establish a general prediction advantage (Appendix~\ref{app:one-triplet-comparisons}).

\subsection{Objective and refinement comparisons}
On the first 20 two-group datasets per correlation, four arms start from the same $K=5$ candidate prepared using the direct objective. Each runs at caps of 16 and 128 refreshes under the same 60-second limit and safeguards, with acceptance based on its own objective. Table~\ref{tab:controlled-kl-contrasts} reports direct joint minus comparator: means and bracketed 95\% paired intervals, with negative values favoring direct joint.

\begin{table}[htbp]
\centering\small\setlength{\tabcolsep}{3pt}
\caption{Direct-joint-minus-comparator KL at both refresh caps.}
\label{tab:controlled-kl-contrasts}
\begin{tabular*}{\linewidth}{@{\extracolsep{\fill}}rlcc@{}}
\toprule
$\rho$ & Comparator & Cap 16 & Cap 128\\
\midrule
0.7 & Augmented joint & $-0.593\;[-0.624,-0.562]$ & $-0.684\;[-0.790,-0.578]$\\
0.7 & Direct frozen & $-0.044\;[-0.060,-0.028]$ & $-0.063\;[-0.086,-0.040]$\\
0.7 & Direct stagewise & $-0.048\;[-0.065,-0.030]$ & $-0.068\;[-0.091,-0.044]$\\
\midrule
0.9 & Augmented joint & $-0.559\;[-0.616,-0.501]$ & $-0.688\;[-0.763,-0.612]$\\
0.9 & Direct frozen & $-0.055\;[-0.072,-0.037]$ & $-0.097\;[-0.130,-0.064]$\\
0.9 & Direct stagewise & $-0.062\;[-0.079,-0.044]$ & $-0.103\;[-0.136,-0.070]$\\
\midrule
0.99 & Augmented joint & $-0.526\;[-0.561,-0.491]$ & $-0.653\;[-0.736,-0.570]$\\
0.99 & Direct frozen & $-0.036\;[-0.048,-0.025]$ & $-0.092\;[-0.140,-0.043]$\\
0.99 & Direct stagewise & $-0.042\;[-0.054,-0.030]$ & $-0.097\;[-0.145,-0.048]$\\
\bottomrule
\end{tabular*}
\end{table}

Direct joint has lower direct posterior KL than augmented joint on all 60 datasets at both caps; all nine mean contrasts per cap favor direct joint. Increasing the cap enlarges every mean advantage, with eight of nine change intervals excluding zero. This supports the direct objective and joint refinement under common initialization. The benefit is metric-dependent: at $\rho=0.99$, grouped-support TV favors augmentation at cap 16 but direct joint at cap 128. Appendix~\ref{app:controlled-local} gives absolute errors, budget sensitivity and stopping outcomes.

\subsection{Complete fitting procedures}
On all 150 two-group datasets, joint, frozen and stagewise fitting share the direct objective, saved MFVI baseline, 60-second search allowance and ten-component cap. Their proposals, search paths and attained sizes can differ, so this comparison assesses complete procedures. Figure~\ref{fig:refinement-effects} shows joint-minus-comparator means and 95\% paired intervals for KL and support TV over the six grouped predictors.

Joint refinement has lower mean KL than frozen refinement at every correlation. Against stagewise fitting, KL and support TV favor joint refinement at $\rho=0.7$. At $\rho=0.9$, KL favors joint refinement, covariance favors stagewise fitting, and TV and PIP intervals include zero. At $\rho=0.99$, stagewise fitting has lower KL, TV, PIP and covariance errors. All prediction intervals between mixture strategies include zero. Thus the controlled local gains do not imply uniform superiority of the complete procedure (Appendix~\ref{app:two-groups}).

\begin{figure}[htbp]
\centering
\includegraphics[width=1.0\linewidth]{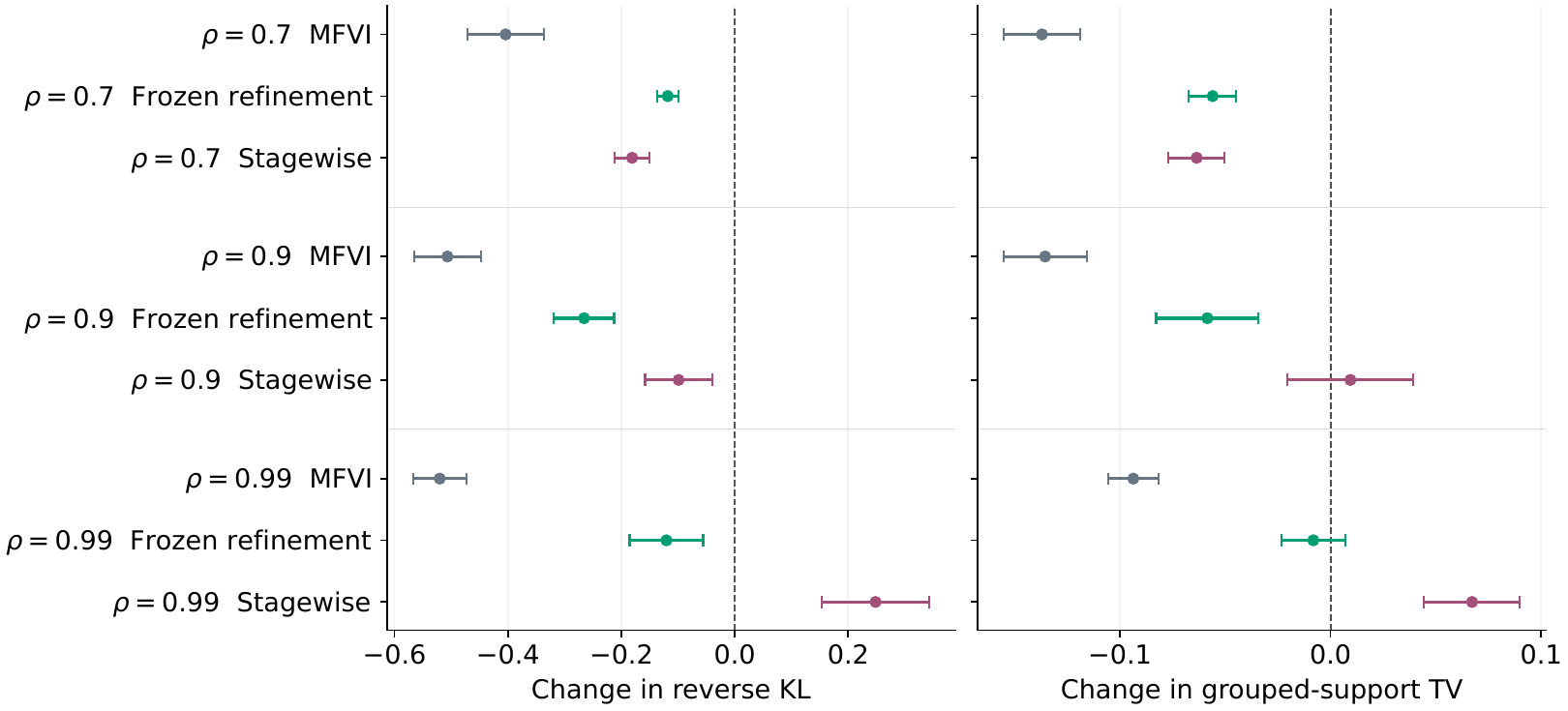}
\caption{KL and grouped-support TV contrasts between fitting strategies.}
\label{fig:refinement-effects}
\end{figure}

\section{Theoretical analysis}
\label{sec:asymptotic}
We relate mixture capacity and optimization error to posterior approximation, contraction, support selection and a Gaussian limit. Proofs are in Appendix~\ref{app:new-theory}.

\subsection{Mixture capacity}
\label{sec:mixture-capacity}
For the nested families in \eqref{eq:mixture-family}, define $a_K=\inf_{Q\in\mathcal Q_K}\KL(Q\|\Pi)$. Infima avoid assuming that a minimizing parameter vector exists; $\mathcal Q_1$ is mean field.

\begin{proposition}[Information bound]
\label{prop:new-nesting}
The families satisfy $\mathcal Q_K\subseteq\mathcal Q_{K+1}$ and $\max\{0,a_1-\log K\}\leq a_K\leq a_{K-1}\leq a_1$ for $K\geq2$. Thus $K$ components improve the optimal reverse KL over mean field by at most $\log K$.
\end{proposition}

For a Gaussian target $N_s(m,A^{-1})$ on a fixed support, the product-Gaussian minimum is $a_1=\log\{\prod_{j=1}^s A_{jj}/\det A\}/2$. If this grows proportionally to $s$, vanishing reverse KL requires $\log K$ to grow at least as quickly. This complements the covariance decomposition in \eqref{eq:mixture-covariance}.

\subsection{Support approximation}
\label{sec:support-approximation}
For a support $S$ with indicator $\one_S$, write $\pi_S=\Pi(\gamma=\one_S\mid y)$. Its active posterior is $N_{|S|}(m_S,A_S^{-1})$, where $A_S=\sigma^{-2}X_S^\top X_S+\tau^{-2}I_{|S|}$ and $m_S=A_S^{-1}\sigma^{-2}X_S^\top y$. Its product-Gaussian error is $d_S=\log\{\prod_{j\in S}(A_S)_{jj}/\det A_S\}/2$, with $d_\varnothing=0$.

\begin{samepage}
\begin{proposition}[Support approximation]
\label{prop:support-capacity}
For a nonempty collection $\mathcal A$ of at most $K$ supports, put $r_{\mathcal A}=\sum_{S\in\mathcal A}\pi_S$ and $\widetilde\pi_S=\pi_S/r_{\mathcal A}$. Then $a_K\leq-\log\{\sum_{S\in\mathcal A}\pi_Se^{-d_S}\}\leq-\log r_{\mathcal A}+\sum_{S\in\mathcal A}\widetilde\pi_Sd_S$. The first bound is the minimum over mixtures with one product Gaussian per retained support, attained by means $m_S$, variances $(A_S)_{jj}^{-1}$ and weights proportional to $\pi_Se^{-d_S}$.
\end{proposition}
\end{samepage}

Within this subclass, $w_S^\star/w_T^\star=(\pi_S/\pi_T)\exp\{-(d_S-d_T)\}$. Thus conditional dependence can distort support weights as well as coefficient uncertainty. The second bound separates omitted mass from retained conditional error: two competing singleton supports of total mass $1-\delta$ give $a_2\leq-\log(1-\delta)$. Large $d_S$ may instead require several components on a support. General mixture components can span many supports; the construction uses deterministic supports and boundary inclusion probabilities. Appendix~\ref{app:capacity-proof} also bounds $d_S$ through normalized precision matrices.

At sample size $n$, define $a_{K,n}=\inf_{Q\in\mathcal Q_K}\KL(Q\|\Pi_n)$. For a computed $Q_n\in\mathcal Q_{K_n}$, the exact decomposition is $\KL(Q_n\|\Pi_n)=a_{K_n,n}+e_n$, where $e_n=\mathcal L_n(Q_n)-\inf_{Q\in\mathcal Q_{K_n}}\mathcal L_n(Q)\geq0$. These quantify family approximation and within-family optimization error, respectively. Appendix~\ref{app:new-theory} gives numerical error bounds and their restricted-domain qualification.

\subsection{Asymptotic results}
\label{sec:statistical-guarantees}
Let $\mathbb P_{\beta^0}=N_n(X\beta^0,\sigma^2I_n)$ be the sampling law for fixed $X$, and write $S_0=\{j:\beta_j^0\ne0\}$, $s_0=|S_0|$ and $S=\{j:\gamma_j=1\}$. Quantities may vary with $n$; $K_n\geq1$ and $Q_n\in\mathcal Q_{K_n}$ may depend on $Y$. All stochastic orders and convergence in probability below refer to $\mathbb P_{\beta^0}$, with numerical seeds fixed.

\begin{assumption}[Sparse regression]
\label{ass:new-contraction}
We have $p\geq n\geq3$, $s_0\geq1$, $s_0\log p=o(n)$ and $\max_j\|X_j\|^2/n\leq C_X$. The known $\sigma^2$ and slab variance $\tau^2$ are bounded above and away from zero, $\|\beta^0\|_\infty\leq B$, and $\omega_n=p^{-a}$, with fixed $C_X,B<\infty$ and $a>11$.
\end{assumption}

\begin{theorem}[Contraction]
\label{thm:new-contraction}
Under Assumption~\ref{ass:new-contraction}, if $e_n=O_{\mathbb P_{\beta^0}}(s_0\log p)$, then $\E_{Q_n}\|X(\beta-\beta^0)\|^2=O_{\mathbb P_{\beta^0}}(\sigma^2s_0\log p)$ and $\E_{Q_n}|S|=O_{\mathbb P_{\beta^0}}(s_0)$. For every $M_n\to\infty$, $Q_n\{\|X(\beta-\beta^0)\|/\sqrt n>\allowbreak M_n\sigma\sqrt{s_0\log p/n}\}\allowbreak\to0$ in $\mathbb P_{\beta^0}$-probability. If integers $m_n$ satisfy $s_0/m_n\to0$ and the sparse eigenvalue condition $\inf_{0<|T|\leq m_n+s_0}\allowbreak\lambda_{\min}(X_T^\top X_T/n)\geq c_X>0$ for a fixed $c_X$, the same conclusion holds with $\|\beta-\beta^0\|$ replacing $\|X(\beta-\beta^0)\|/\sqrt n$.
\end{theorem}

\begin{corollary}[Baseline transfer]
\label{cor:baseline-contraction}
Under Assumption~\ref{ass:new-contraction}, let $Q_n^\circ$ be a possibly data-dependent baseline, with both $Q_n^\circ$ and $Q_n$ having finite objectives. If nonnegative random $r_n,\delta_n$ satisfy $\mathcal L_n(Q_n^\circ)-\ell_n(\beta^0)\leq r_n$, $\mathcal L_n(Q_n)-\mathcal L_n(Q_n^\circ)\leq\delta_n$ and $r_n+\delta_n=O_{\mathbb P_{\beta^0}}(s_0\log p)$, the prediction and model-size conclusions of Theorem~\ref{thm:new-contraction} hold, as does coefficient contraction under its sparse eigenvalue condition. Here $\ell_n$ is the loss at sample size $n$.
\end{corollary}

Theorem~\ref{thm:new-contraction} controls prediction error and bounds expected model size at the order of the true sparsity; the sparse eigenvalue condition also gives coefficient contraction. This rate is shared with mean field. Corollary~\ref{cor:baseline-contraction} shows that comparison with a suitably controlled baseline suffices to retain it. The next results concern support recovery and distributional approximation. Write $\pi_{*,n}=\Pi_n(S=S_0\mid Y)$ and recall $A_{S_0}=\sigma^{-2}X_{S_0}^\top X_{S_0}+\tau^{-2}I_{s_0}$.

\begin{theorem}[Selection consistency]
\label{thm:new-selection}
Suppose $p\to\infty$, $s_0\geq1$ and $\mathbb P_{\beta^0}\{1-\pi_{*,n}\leq p^{-b}\}\to1$ for fixed $b>0$. If the condition numbers of $A_{S_0}$ are uniformly bounded and $s_0+e_n=o_{\mathbb P_{\beta^0}}(\log p)$, then $Q_n(S=S_0)\to1$ in $\mathbb P_{\beta^0}$-probability. The thresholded support $\widehat S=\{j:Q_n(\gamma_j=1)>1/2\}$ satisfies $\mathbb P_{\beta^0}(\widehat S=S_0)\to1$.
\end{theorem}

Under Theorem~\ref{thm:new-selection}'s conditions, thresholding the variational PIPs at $1/2$ consistently recovers $S_0$. For the Gaussian approximation, put $B_n=\sigma^{-2}X_{S_0}^\top X_{S_0}$ and $\widehat\beta_{S_0}=(X_{S_0}^\top X_{S_0})^{-1}X_{S_0}^\top Y$. Let $\mathcal G_n$ fix the support at $S_0$, set inactive coefficients to zero and give active coefficients the law $N_{s_0}(\widehat\beta_{S_0},B_n^{-1})$.

\begin{samepage}
\begin{theorem}[Bernstein--von Mises]
\label{thm:new-bvm}
Suppose $s_0\geq1$, $\pi_{*,n}\to1$ in $\mathbb P_{\beta^0}$-probability, the eigenvalues of $B_n$ lie in $[cn,Cn]$ for fixed $0<c\leq C<\infty$, and $\tau^{-2}$ is bounded. If $s_0/n\to0$, $\|\beta_{S_0}^0\|^2/n\to0$ and $a_{K_n,n}+e_n=o_{\mathbb P_{\beta^0}}(1)$, then $\TV(Q_n,\mathcal G_n)\to0$ in $\mathbb P_{\beta^0}$-probability. The law of $\sqrt n(\beta_{S_0}-\widehat\beta_{S_0})$ under $Q_n$ has total-variation distance from $N_{s_0}(0,nB_n^{-1})$ tending to zero in $\mathbb P_{\beta^0}$-probability.
\end{theorem}
\end{samepage}

Theorem~\ref{thm:new-bvm} gives a Gaussian approximation to the joint law of $\beta_{S_0}$ under $Q_n$. It is centered at the least-squares estimator on $S_0$ and is asymptotically accurate uniformly over measurable events.

\section{Discussion}
\label{sec:discussion}
The improvements in support probabilities and coefficient covariance show why dependence matters when several sparse explanations are plausible. Similar predictions can arise from different inclusion patterns, so better posterior approximation need not reduce prediction error. Posterior comparisons therefore reveal differences that predictive performance alone may miss.

Local refinement and complete fitting address different questions. Direct joint refinement lowers KL from the tested common initializations, but complete stagewise fitting is more accurate near collinearity. This contrast motivates examining proposals and search paths alongside refinement. The theory separates approximation from optimization error; the support-based bound further highlights coverage of plausible models and dependence within them. These distinctions suggest possible sources of error without identifying what limits the fitted mixtures.

Two directions follow. Crossing proposal strategies with joint and stagewise refinement under matched budgets would clarify their separate and combined effects. Varying mixture size and refinement effort independently would help assess whether remaining errors respond more to added capacity or further optimization.

\label{preprint:mainend}
\section*{AI use statement}
AI tools assisted with writing, proofs, simulation code, execution and result interpretation. The authors are responsible for all methods, results and conclusions.

\bibliographystyle{abbrvnat}
\bibliography{references}
\appendix

\section{Mixture calculations}
\label{app:mixture}

\subsection{Objective and moments}
\begin{proof}[Proof of Proposition~\ref{prop:mixture-information}]
Under the joint law of $(C,Z)$, let $q_k=\mathrm{d}Q_k/\mathrm{d}\nu$ and $q_\theta=\sum_kw_kq_k$. The relative-entropy chain rule gives
\[
 \sum_kw_k\KL(Q_k\|P)=\KL(Q_\theta\|P)+\sum_kw_k\KL(Q_k\|Q_\theta).
\]
The second term is $\mathcal J(\theta)$. Expected loss is linear in the mixture weights, which proves the stated objective identity. For each positive-weight component, $q_\theta\geq w_kq_k$, so $\KL(Q_k\|Q_\theta)\leq-\log w_k$. Averaging yields $\mathcal J(\theta)\leq H(w)\leq\log K$, while nonnegativity follows from relative entropy. Zero-weight components contribute nothing. The argument applies to the common mixed measure and requires no Lebesgue density at the inactive point.
\end{proof}

Each component objective is
\begin{align}
 \mathcal L_k={}&\frac{\|y-Xm_k\|^2+\sum_j\|X_j\|^2d_{kj}}{2\sigma^2}
 +\sum_j\left\{\alpha_{kj}\log\frac{\alpha_{kj}}{\omega}+(1-\alpha_{kj})\log\frac{1-\alpha_{kj}}{1-\omega}\right\}\notag\\
 &+\frac12\sum_j\alpha_{kj}\left\{\frac{v_{kj}+\mu_{kj}^2}{\tau^2}-1+\log\frac{\tau^2}{v_{kj}}\right\}.
 \label{eq:mixture-component-cost}
\end{align}

The convention $0\log0=0$ applies at boundary inclusion probabilities.

The component moments are $\E_{Q_k}\beta=m_k$ and $\operatorname{Cov}_{Q_k}(\beta)=\operatorname{diag}(d_k)$. Expanding the squared residual gives the loss in \eqref{eq:mixture-component-cost}. For each coordinate,
\[
 \KL(Q_{kj}\|P_j)=\KL\{\operatorname{Bernoulli}(\alpha_{kj})\|\operatorname{Bernoulli}(\omega)\}+\alpha_{kj}\KL\{N(\mu_{kj},v_{kj})\|N(0,\tau^2)\}.
\]
The Gaussian divergence formula gives the remaining terms. Conditioning on $C$ proves the covariance decomposition \eqref{eq:mixture-covariance} and joint inclusion formula in Section~\ref{sec:mixture-family}. The vectors $\sqrt{w_k}(m_k-\bar m)$ have a linear dependence, so their covariance contribution has rank at most $K-1$.

\subsection{Augmented objectives}
\label{app:augmentation}
For the representation in Section~\ref{sec:direct-augmentation}, let $T(\gamma,\beta^+)=(\gamma,\gamma\odot\beta^+)$ and $Q=T_\#Q^+$. On a support $S$, conditioning on $z=(\gamma,\beta)$ fixes $\beta_S^+=\beta_S$. The likelihood is independent of $\beta_{S^c}^+$, so the augmented posterior conditional law of these discarded variables is $\mathcal N_{S^c}$. Applying the relative-entropy chain rule to $T$ gives the decomposition in Section~\ref{sec:direct-augmentation}. Under a product law with independent indicators and latent Gaussians, each inactive coordinate contributes its Gaussian prior divergence, and averaging over indicators gives the displayed product penalty.

For any $Q$, attach independent prior draws to its inactive coordinates to construct an augmented law with zero conditional KL, hence $\inf_{Q^+:T_\#Q^+=Q}\KL(Q^+\|\Pi^+)=\KL(Q\|\Pi)$.
For a $K$-component mixture on $(\gamma,\beta)$ this completion can keep the same label and coordinate independence conditional on that label, but generally makes each latent Gaussian depend on its indicator. The extra cost is induced by independence between the indicator and latent Gaussian variable. The consequences of auxiliary independence are also discussed by \citet{spence2020flexible}. Their proposed method uses product spike-and-slab factors, while their separate boosting baseline uses augmented Bernoulli/Gaussian mixtures. The released boosting code uses stagewise residual-ELBO updates with a prescribed new-component weight schedule.\footnote{\href{https://github.com/jeffspence/non_overlapping_mixtures/blob/main/code/pyro_code_discrete.py}{Released implementation of the augmented boosting baseline}.}

\subsection{Mixture derivatives}
\label{app:mixture-gradients}
Work at positive weights and variances and interior inclusion probabilities. Differentiation under the integrals is justified on compact local parameter sets with mixture weights bounded away from zero, inclusion probabilities bounded away from zero and one, and variances bounded away from zero. Set $\eta_{kj}=\log\{\alpha_{kj}/(1-\alpha_{kj})\}$ and $t_{kj}=\log v_{kj}$. The component scores are $\partial_{\eta_{kj}}\log q_k(z)=\gamma_j-\alpha_{kj}$, $\partial_{\mu_{kj}}\log q_k(z)=\gamma_j(\beta_j-\mu_{kj})/v_{kj}$, and $\partial_{t_{kj}}\log q_k(z)=(\gamma_j/2)\{(\beta_j-\mu_{kj})^2/v_{kj}-1\}$. The Gaussian scores vanish at the inactive point, but the inclusion score does not. Thus the inactive contribution must remain in the integral.

Let $s_{kj}$ denote one of these scores and $\vartheta_{kj}$ its parameter, and put $D_k=\KL(Q_k\|Q_\theta)$. For softmax weight logits $u_k$,
\begin{equation*}
 \partial_{\vartheta_{kj}}\mathcal J(\theta)=w_k\E_{Q_k}\left[s_{kj}(Z)\log\frac{q_k(Z)}{q_\theta(Z)}\right],\qquad \partial_{u_k}\mathcal J(\theta)=w_k(D_k-\mathcal J(\theta)).
\end{equation*}
Indeed, differentiating $\mathcal J(\theta)=\sum_kw_k\int q_k\log q_k\,\mathrm{d}\nu-\int q_\theta\log q_\theta\,\mathrm{d}\nu$ gives the first identity because the normalization terms cancel. The unconstrained weight derivative is $D_k-1$, and the softmax derivative gives the second identity. Combined with the analytic component objectives, these are population derivatives of \eqref{eq:objective}.

A finite importance sum need not integrate a normalized density to exactly one. For the implemented mixture-proposal rule, let $a=(h,b)$ index a fixed reference point, put $\xi_a=\widetilde w_h/N$, and define $B_{ak}=\xi_aw_kq_k(z_a)/r(z_a)$, $R_{ak}=\log\{q_k(z_a)/q_\theta(z_a)\}$, and $\widehat M=\sum_{a,k}B_{ak}$. Then $\widehat{\mathcal J}(\theta)=\sum_{a,k}B_{ak}R_{ak}$. Direct differentiation, holding $\xi_a$, $r$, and $z_a$ fixed, gives
\begin{align}
 \partial_{\vartheta_{kj}}\widehat{\mathcal J}(\theta)&=\sum_aB_{ak}R_{ak}s_{kj}(z_a),\label{eq:finite-information-gradient}\\
 \partial_{u_k}\widehat{\mathcal J}(\theta)&=\sum_aB_{ak}(R_{ak}-1)-w_k(\widehat{\mathcal J}(\theta)-\widehat M).\notag
\end{align}
For the component derivative, the extra terms cancel pointwise because every component is included in the inner sum. For the weight derivative, the finite estimated masses need not equal their exact values, so the terms involving $\widehat M$ and $\sum_aB_{ak}$ must be retained. These formulas differentiate precisely the numerical objective supplied to L-BFGS-B, including the changing importance ratios. Finite differences check the derivatives, and independent integration draws assess numerical integration variability.

\subsection{Component gradients}
\begingroup\allowdisplaybreaks[1]
Write $\widetilde r_k=Xm_k-y$, $g_k=X^\top \widetilde r_k/\sigma^2$, $h_j=\|X_j\|^2/\sigma^2$, and $A_{kj}=(1/2)\{(v_{kj}+\mu_{kj}^2)/\tau^2-1+\log(\tau^2/v_{kj})\}$. Differentiation of \eqref{eq:mixture-component-cost} yields
\begin{align*}
 \partial_{\alpha_{kj}}\mathcal L_k
 &=g_{kj}\mu_{kj}+\tfrac12h_j\{v_{kj}+(1-2\alpha_{kj})\mu_{kj}^2\}
 +\log\frac{\alpha_{kj}(1-\omega)}{(1-\alpha_{kj})\omega}+A_{kj},\\
 \partial_{\mu_{kj}}\mathcal L_k
 &=\alpha_{kj}g_{kj}+h_j\alpha_{kj}(1-\alpha_{kj})\mu_{kj}+\alpha_{kj}\mu_{kj}/\tau^2,\\
 \partial_{v_{kj}}\mathcal L_k
 &=\tfrac12\alpha_{kj}(h_j+\tau^{-2}-v_{kj}^{-1}).
\end{align*}
\endgroup
The logit and log-variance derivatives follow by multiplying by $\alpha_{kj}(1-\alpha_{kj})$ and $v_{kj}$, respectively. With $\overline{\mathcal L}=\sum_kw_k\mathcal L_k$, the weight-logit derivative of the analytic component term is $w_k(\mathcal L_k-\overline{\mathcal L})$. Subtracting \eqref{eq:finite-information-gradient} completes the implemented finite-objective gradient.

\subsection{Acceptance bounds}
\label{app:mixture-acceptance}
Suppose numerical objective values have valid absolute error bounds $|\widehat{\mathcal L}(Q)-\mathcal L(Q)|\leq\epsilon$ and $|\widehat{\mathcal L}(Q')-\mathcal L(Q')|\leq\epsilon'$. For a requested decrease $\delta\geq0$, the rule $\widehat{\mathcal L}(Q')+\epsilon'<\widehat{\mathcal L}(Q)-\epsilon-\delta$ certifies exact descent, since
\[
 \mathcal L(Q')\leq\widehat{\mathcal L}(Q')+\epsilon'<\widehat{\mathcal L}(Q)-\epsilon-\delta\leq\mathcal L(Q)-\delta.
\]
Applying this argument successively gives the corresponding incumbent guarantee. The empirical scramble standard errors in \eqref{eq:empirical-acceptance} do not alone supply these error bounds, particularly when candidates are compared adaptively.

\section{Theoretical results}
\label{app:new-theory}

If a numerical objective $\widehat{\mathcal L}_n$ has uniform error at most $\epsilon_n$ over $\mathcal Q_{K_n}$, and its fitted value is within $\delta_{\mathrm{opt},n}$ of its global infimum, then $e_n\leq\delta_{\mathrm{opt},n}+2\epsilon_n$. If these bounds hold only on a restricted parameter domain, its infimum gap relative to $\mathcal Q_{K_n}$ must also be included in $e_n$. Local convergence and finitely many validation checks do not certify these global bounds.

\begin{proof}[Proof of Proposition~\ref{prop:new-nesting}]
Appending a zero-weight component gives $\mathcal Q_K\subseteq\mathcal Q_{K+1}$. For any mixture $Q=Q_\theta$, Proposition~\ref{prop:mixture-information} and the fact that $\mathcal L(Q)=\KL(Q\|\Pi)-\log\mathcal Z$ give
\[
 \KL(Q\|\Pi)=\sum_{k=1}^K w_k\KL(Q_k\|\Pi)-\mathcal J(\theta)
 \geq a_1-H(w)\geq a_1-\log K.
\]
Nonnegativity of KL and taking infima prove the claimed bounds. For completeness, if the target is $N_s(m,A^{-1})$, a product Gaussian with mean $b$ and diagonal covariance $D$ has divergence
\[
 \frac12\{(b-m)^\top A(b-m)+\operatorname{tr}(AD)-s-\log\det A-\log\det D\}.
\]
Its minimizer is $b=m$ and $D_{jj}=A_{jj}^{-1}$, proving the Gaussian mean-field formula in Section~\ref{sec:mixture-capacity}. Finally, the numerical error bound follows from
\[
 \mathcal L_n(Q_n)
 \leq\widehat{\mathcal L}_n(Q_n)+\epsilon_n
 \leq\inf_{Q\in\mathcal Q_{K_n}}\widehat{\mathcal L}_n(Q)+\delta_{\mathrm{opt},n}+\epsilon_n
 \leq\inf_{Q\in\mathcal Q_{K_n}}\mathcal L_n(Q)+\delta_{\mathrm{opt},n}+2\epsilon_n.\qedhere
\]
\end{proof}

\subsection{Capacity bound}\label{app:capacity-proof}
\begin{proof}[Proof of Proposition~\ref{prop:support-capacity}]
The Gaussian likelihood and slab prior give the stated conditional posterior by completing the square. The positive slab precision makes $A_S$ positive definite even when $X_S$ is rank deficient. The support masses are strictly positive under $0<\omega<1$ and positive finite $\sigma^2,\tau^2$.

Let $G_S$ be a product Gaussian with mean $b_S$ and positive diagonal covariance $V_S$ on support $S$, with inactive coefficients fixed at zero. Its conditional reverse KL is
\[
 c_S(G_S)=\frac12\left\{
 (b_S-m_S)^\top A_S(b_S-m_S)+\operatorname{tr}(A_SV_S)
 -|S|-\log\det A_S-\log\det V_S\right\}.
\]
It is minimized at $b_S=m_S$ and $(V_S)_{jj}=(A_S)_{jj}^{-1}$, with minimum $d_S$. The empty support has a unique conditional law and conditional KL zero. Hadamard's inequality gives $d_S\geq0$.

Distinct supports are disjoint events on the joint state space, so the KL chain rule for $Q=\sum_{S\in\mathcal A}w_SG_S$ gives
\[
 \KL(Q\|\Pi)
 =\sum_{S\in\mathcal A}w_S\log\frac{w_S}{\pi_S}
  +\sum_{S\in\mathcal A}w_Sc_S(G_S).
\]
Use the conditional minimizers, and let $B_{\mathcal A}=\sum_{S\in\mathcal A}\pi_Se^{-d_S}$. Define $w_S^*=\pi_Se^{-d_S}/B_{\mathcal A}$. For every probability vector $w$ on $\mathcal A$,
\[
 \sum_{S\in\mathcal A}w_S\left(\log\frac{w_S}{\pi_S}+d_S\right)
 =\KL(w\|w^*)-\log B_{\mathcal A}.
\]
This proves both the optimal weights and equality for the stated one-product-per-support subclass. Adding zero-weight components embeds this mixture in $\mathcal Q_K$. Finally,
\[
 -\log B_{\mathcal A}
 =-\log r_{\mathcal A}
  -\log\sum_{S\in\mathcal A}\widetilde\pi_Se^{-d_S}
 \leq-\log r_{\mathcal A}
       +\sum_{S\in\mathcal A}\widetilde\pi_Sd_S,
\]
by Jensen's inequality. This completes the proof.
\end{proof}

\paragraph{Precision bound.}
For a nonempty support let $D_S=\operatorname{diag}(A_S)$ and $E_S=D_S^{-1/2}A_SD_S^{-1/2}-I_{|S|}$. If $\|E_S\|_{\mathrm{op}}\leq\eta<1$, then
\begin{equation*}
 \frac{\|E_S\|_F^2}{4(1+\eta)}
 \leq d_S
 \leq\frac{\|E_S\|_F^2}{4(1-\eta)}
 \leq\frac{|S|\eta^2}{4(1-\eta)}.
\end{equation*}
Indeed, $d_S=-(1/2)\log\det(I_{|S|}+E_S)$ and $\operatorname{tr}(E_S)=0$. If $\lambda_1,\ldots,\lambda_{|S|}$ are the eigenvalues of $E_S$, then
\[
 d_S=\frac12\sum_i\{\lambda_i-\log(1+\lambda_i)\}
     =\frac12\sum_i\lambda_i^2
       \int_0^1\frac{t}{1+t\lambda_i}\,\mathrm{d}t.
\]
Bounding the denominator by $1-\eta$ and $1+\eta$ proves the result. Set $\|E_\varnothing\|_F=0$ for the empty support. Consequently, if this operator-norm condition holds uniformly over the retained supports with a common fixed $\eta<1$, then
\[
 a_K\leq-\log r_{\mathcal A}
       +\frac{1}{4(1-\eta)}
        \sum_{S\in\mathcal A}\widetilde\pi_S\|E_S\|_F^2.
\]
The Frobenius norm of the normalized precision perturbation accounts for dependence accumulated across active coordinates as support dimension grows.

\subsection{Proof of contraction}\label{app:contraction-statement}
We use the notation and assumptions of Section~\ref{sec:statistical-guarantees}.

\begin{proof}[Proof of Theorem~\ref{thm:new-contraction}]
Write $s=s_0$ and $\varepsilon=Y-X\beta^0$, and let $H_T$ be the orthogonal projection onto the span of $X_T$. For a chi-square variable $V$ with $r$ degrees of freedom, its moment generating function gives $\Pr\{V>r+2\sqrt{rt}+2t\}\leq e^{-t}$. For a support $S$ of size $k$, apply this inequality with $r\leq k+s$ and $t=(k+s)\log p$. A union bound gives an event $\mathcal E_n$, with $\mathbb P_{\beta^0}(\mathcal E_n^c)\leq p^{-s}(1+p^{-1})^p\leq ep^{-s}$, on which $\|H_{S\cup S_0}\varepsilon\|^2/\sigma^2\leq10(|S|+s)\log p$ simultaneously for all supports. This includes rank-deficient supports. Expanding the squared loss and applying Cauchy--Schwarz and $uv\leq u^2/4+v^2$ yield, on this event,
\begin{equation}
 \ell(\beta)-\ell(\beta^0)
 \geq\frac{\|X(\beta-\beta^0)\|^2}{4\sigma^2}-10(|S|+s)\log p.
 \label{eq:new-proof-loss}
\end{equation}

Let $Q_S$ be the support law induced by an arbitrary candidate $Q$. Comparison with the probability mass function $\rho(S)=p^{-|S|}/(1+p^{-1})^p$ bounds its entropy by $H(Q_S)\leq(\log p)\E_Q|S|+1$. Let $P_S$ be the Bernoulli support prior, with probability mass function $p_S$. Then $-\log p_S(S)\geq a|S|\log p$. Relative-entropy data processing therefore gives
\begin{equation}
 \KL(Q\|P)\geq\KL(Q_S\|P_S)\geq(a-1)(\log p)\E_Q|S|-1.
 \label{eq:new-proof-modelsize}
\end{equation}

As comparator, take the product law $Q^0$ with support fixed at $S_0$, active coefficients independently distributed as $N(\beta_j^0,\sigma^2/n)$, and inactive coefficients zero. It belongs to $\mathcal Q_1\subseteq\mathcal Q_{K_n}$ for every realization of $K_n$. Its centered expected loss is $(2n)^{-1}\sum_{j\in S_0}\|X_j\|^2\leq C_Xs/2$. Its prior divergence is
\begin{align*}
 \KL(Q^0\|P)={}&as\log p-(p-s)\log(1-p^{-a})\\
 &+\frac12\sum_{j\in S_0}\left\{\frac{\sigma^2/n+(\beta_j^0)^2}{\tau^2}-1+\log\frac{n\tau^2}{\sigma^2}\right\}.
\end{align*}
Assumption~\ref{ass:new-contraction}, $p\geq n$, and $-(p-s)\log(1-p^{-a})\leq2p^{1-a}$ eventually show that $\mathcal L(Q^0)-\ell(\beta^0)\leq C_0s\log p$. By the definition of $e_n$, $\mathcal L(Q_n)-\ell(\beta^0)\leq C_0s\log p+e_n$. Integrating \eqref{eq:new-proof-loss} and using \eqref{eq:new-proof-modelsize} now gives on $\mathcal E_n$
\begin{equation*}
 \frac{\E_{Q_n}\|X(\beta-\beta^0)\|^2}{4\sigma^2}
 +(a-11)(\log p)\E_{Q_n}|S|
 \leq(C_0+10)s\log p+e_n+1.
\end{equation*}
Because $a>11$, $e_n=O_{\mathbb P_{\beta^0}}(s\log p)$, and $\mathbb P_{\beta^0}(\mathcal E_n)\to1$, this proves both moment bounds. Markov's inequality proves prediction contraction.

Under the sparse eigenvalue condition in Theorem~\ref{thm:new-contraction}, on $\{|S|\leq m_n\}$ we have $\|X(\beta-\beta^0)\|^2/n\geq c_X\|\beta-\beta^0\|^2$. Hence the coefficient tail probability is bounded by $Q_n(|S|>m_n)+O_{\mathbb P_{\beta^0}}(M_n^{-2})=O_{\mathbb P_{\beta^0}}(s/m_n)+O_{\mathbb P_{\beta^0}}(M_n^{-2})$, which tends to zero in $\mathbb P_{\beta^0}$-probability. This proves coefficient contraction without a beta-min condition.
\end{proof}

\subsection{Proof of baseline transfer}\label{app:baseline-transfer}

\begin{proof}[Proof of Corollary~\ref{cor:baseline-contraction}]
The event $\mathcal E_n$ in the proof of Theorem~\ref{thm:new-contraction} does not depend on a candidate distribution. Integrating \eqref{eq:new-proof-loss} and applying
\eqref{eq:new-proof-modelsize} gives, for every finite-objective $Q$,
\[
 \frac{\E_Q\|X(\beta-\beta^0)\|^2}{4\sigma^2}
 +(a-11)(\log p)\E_Q|S|
 \leq \mathcal L_n(Q)-\ell_n(\beta^0)
       +10s_0\log p+1.
\]
Substituting $Q=Q_n$ and the assumed objective bounds gives the upper bound
$r_n+\delta_n+10s_0\log p+1$. Since $a>11$ and $\mathbb P_{\beta^0}(\mathcal E_n)\to1$, the two moment bounds follow. The same Markov and sparse-eigenvalue arguments used in the theorem give the tail conclusions.
\end{proof}

The one-sided inequalities permit any improvement over the baseline. The truth-centered product distribution in the theorem's proof satisfies the first inequality with $r_n=C_0s_0\log p$. A computable baseline satisfying the same bound also yields contraction, provided the comparison error is controlled at the stated order.

\subsection{Posterior transfer}\label{app:transfer-statements}

Suppose a nonempty support $S_0$ has posterior probability $\pi_*\in(0,1)$ and conditional active-coefficient posterior $N_{s_0}(m_*,A_*^{-1})$, where $A_*\succ0$. This conditional distribution is Gaussian exactly under our model. Its minimum reverse KL over product Gaussians is $d_* = (1/2)\log\{\prod_{j\in S_0}(A_*)_{jj}/\det A_*\}$.

\begin{proposition}[Posterior transfer]
\label{prop:new-transfer}
For any $Q\in\mathcal Q_K$ with within-family error at most $e\geq0$,
\begin{equation}
 Q(S\ne S_0)
 \leq\frac{-\log\pi_*+d_*+e+\log2}{\log\{1/(1-\pi_*)\}}.
 \label{eq:new-selection-transfer}
\end{equation}
More generally, for an event $B$ with $0<\Pi(B)<1$,
\begin{equation*}
 Q(B)\leq\frac{a_K+e+\log2}{\log\{1/\Pi(B)\}},\qquad
 \TV(Q,\Pi)\leq\sqrt{(a_K+e)/2}.
\end{equation*}
If the condition number of $A_*$ is at most $\kappa$, then $0\leq d_*\leq s_0\log\kappa/2$.
\end{proposition}

\begin{proof}[Proof of Proposition~\ref{prop:new-transfer}]
The product Gaussian with mean $m_*$ and covariance $\{\operatorname{diag}(A_*)\}^{-1}$, supported on $S_0$, belongs to $\mathcal Q_1$. Its KL to the full posterior is $-\log\pi_*+d_*$ by conditioning on the support and the Gaussian mean-field formula in Section~\ref{sec:mixture-capacity}. Comparison with this candidate and the within-family error bound give $\KL(Q\|\Pi)\leq-\log\pi_*+d_*+e$. For any event $B$, write $t=Q(B)$ and $u=\Pi(B)$. Binary data processing and the entropy bound $-t\log t-(1-t)\log(1-t)\leq\log2$ give
\[
 \KL(Q\|\Pi)\geq t\log(t/u)+(1-t)\log\{(1-t)/(1-u)\}\geq t\log(1/u)-\log2.
\]
Taking $B=\{S\ne S_0\}$ proves \eqref{eq:new-selection-transfer}. Using $\KL(Q\|\Pi)\leq a_K+e$ proves the event-probability bound, and Pinsker's inequality proves the total-variation bound. Hadamard's inequality gives $d_*\geq0$, while $\prod_j(A_*)_{jj}\leq\lambda_{\max}(A_*)^{s_0}$ and $\det A_*\geq\lambda_{\min}(A_*)^{s_0}$ give $d_*\leq s_0\log\kappa/2$.
\end{proof}

\begin{proof}[Proof of Theorem~\ref{thm:new-selection}]
Apply Proposition~\ref{prop:new-transfer} with $\pi_*=\pi_{*,n}$ and $A_*=A_{S_0}$. On an event with $\mathbb P_{\beta^0}$-probability tending to one, the denominator in \eqref{eq:new-selection-transfer} is at least $b\log p$. Moreover, $-\log\pi_{*,n}=o_{\mathbb P_{\beta^0}}(1)$, and bounded condition numbers give $d_*=O(s_0)$. Since $s_0+e_n=o_{\mathbb P_{\beta^0}}(\log p)$ and $p\to\infty$, the numerator is $o_{\mathbb P_{\beta^0}}(\log p)$, proving $Q_n(S\ne S_0)\xrightarrow{\mathbb P_{\beta^0}}0$. Whenever $Q_n(S=S_0)>1/2$, every true coordinate has PIP above $1/2$ and every inactive coordinate has PIP below $1/2$. Hence $\widehat S=S_0$ with $\mathbb P_{\beta^0}$-probability tending to one.
\end{proof}

\begin{proof}[Proof of Theorem~\ref{thm:new-bvm}]
Let $r_n=\tau^{-2}$ and $A_n=B_n+r_nI_{s_0}$. The exact conditional posterior on $S_0$ is $G_n=N_{s_0}(A_n^{-1}B_n\widehat\beta_{S_0},A_n^{-1})$. Write $G_n^0=N_{s_0}(\widehat\beta_{S_0},B_n^{-1})$. If $b_i$ are the eigenvalues of $B_n$, the covariance contribution to $2\KL(G_n\|G_n^0)$ is
\[
 \sum_{i=1}^{s_0}\left\{\frac1{1+r_n/b_i}-1+\log(1+r_n/b_i)\right\}\leq\frac12\sum_{i=1}^{s_0}(r_n/b_i)^2=O(s_0/n^2).
\]
The inequality follows by integrating $t/(1+t)^2\leq t$ from zero. The difference in means is $-r_nA_n^{-1}\widehat\beta_{S_0}$, whose squared $B_n$-norm is bounded by $C'\|\widehat\beta_{S_0}\|^2/n$. Under the sampling model, $\widehat\beta_{S_0}-\beta_{S_0}^0\sim N_{s_0}(0,B_n^{-1})$, so $\|\widehat\beta_{S_0}\|^2=O_{\mathbb P_{\beta^0}}(\|\beta_{S_0}^0\|^2+s_0/n)$. The assumptions imply $\KL(G_n\|G_n^0)=o_{\mathbb P_{\beta^0}}(1)$ and hence $\TV(G_n,G_n^0)=o_{\mathbb P_{\beta^0}}(1)$.

On the joint state space, the distance from $\Pi_n$ to the law supported on $S_0$ with conditional distribution $G_n$ is exactly $1-\Pi_n(S=S_0\mid Y)$. Thus the triangle inequality and Proposition~\ref{prop:new-transfer} give
\[
 \TV(Q_n,\mathcal G_n)
 \leq\sqrt{(a_{K_n,n}+e_n)/2}+1-\Pi_n(S=S_0\mid Y)+\TV(G_n,G_n^0)
 =o_{\mathbb P_{\beta^0}}(1).
\]
This also implies $Q_n(S=S_0)\xrightarrow{\mathbb P_{\beta^0}}1$. For each realization of $Y$, the measurable map $T_n(z)=\sqrt n(\beta_{S_0}-\widehat\beta_{S_0})$ sends $\mathcal G_n$ to $N_{s_0}(0,nB_n^{-1})$. Total variation cannot increase under this map, so $\TV(Q_n\circ T_n^{-1},N_{s_0}(0,nB_n^{-1}))\leq\TV(Q_n,\mathcal G_n)=o_{\mathbb P_{\beta^0}}(1)$, proving the final assertion without conditioning on $S$.
\end{proof}

\section{Experimental details}
\label{app:numerical}

\subsection{Design and metrics}
The designs in Table~\ref{tab:designs} use independent standard Gaussian factors. Within each group, $\widetilde X_j=\sqrt\rho U+\sqrt{1-\rho}\epsilon_j$, with a separate common factor $U$ per group; remaining columns are independent. One uniformly chosen member of each group is active, together with four or three independent columns in the one- or two-group design. In this order the signals are $(0.7,-0.7,0.7,-0.7,0.7)$. We permute columns, center and scale using training means and standard deviations with divisor $n$, and generate $y=X\beta^0+\epsilon$ with $\epsilon\sim N(0,I_n)$. The same transformations apply to the independent test rows.

\begin{table}[htbp]
\centering\small\setlength{\tabcolsep}{4pt}
\caption{Simulation designs: 50 independent datasets per $(p,\rho)$ cell.}
\label{tab:designs}
\begin{tabular*}{\linewidth}{@{\extracolsep{\fill}}lrrlr@{}}
\toprule
Groups (3 predictors) & $p$ & Signals outside groups & $\rho$ & Datasets\\
\midrule
One & $10,20,30,100$ & 4 & $0.7,0.9$ & 400\\
Two & 10 & 3 & $0.7,0.9,0.99$ & 150\\
\bottomrule
\end{tabular*}
\end{table}

All 550 datasets enter the primary comparisons. The local objective/refinement study reuses 60 two-group datasets for 240 fits at each refresh cap. The one-group fitting and time-budget checks reuse 40 and 20 datasets, respectively. Settings and indices were fixed before fitting the comparison arms. Truth, group labels and reference posteriors are used only for evaluation.

For $\pi_j=\Pi(\gamma_j=1\mid y)$ and $\widehat\pi_j=Q(\gamma_j=1)$, PIP error on $A$ is $|A|^{-1}\sum_{j\in A}|\widehat\pi_j-\pi_j|$. Unless stated otherwise, it uses all predictors. Grouped-support TV is $(1/2)\sum_{b\in\{0,1\}^{|B|}}|Q(\gamma_B=b)-\Pi(\gamma_B=b\mid y)|$ on the three or six correlated predictors $B$. Mixture pattern probabilities are analytic. Covariance error is $\|\operatorname{Cov}_Q(\beta)-\operatorname{Cov}_\Pi(\beta\mid y)\|_F$, with $\|\cdot\|_F$ the Frobenius norm. Prediction MSE is $\|X^\dagger(\E_Q\beta-\beta^0)\|^2/1000$, where $X^\dagger$ contains the 1,000 test rows. At $p=10$, all 1,024 supports are enumerated for exact normalizers and moments; larger dimensions use the MCMC references below.

For $R$ paired dataset differences $D_r$, intervals are $\overline D\pm t_{R-1,0.975}s_D/\sqrt R$, where $s_D$ is their sample standard deviation and $t_{R-1,0.975}$ is a Student-$t$ quantile. No multiplicity adjustment is applied. Independent integration uses four scrambles of 16,384 points per component. Integration variability and reference error are assessed separately and are not fully incorporated into these descriptive intervals. Tables show means, with bracketed intervals when available; smaller errors are better. No outcome is excluded by method performance.

\subsection{Posterior approximation}
\label{app:one-triplet-comparisons}
Tables~\ref{tab:formal-posterior} and~\ref{tab:formal-paired} use 50 datasets per $(p,\rho)$ cell, restricting PIP and support errors to eligible references; $R$ counts these references. Objective and prediction comparisons retain every dataset. Write $\Delta_{\mathrm{MF}}=\mathcal L(Q)-\mathcal L(Q_{\mathrm{MF}})$ for the mixture-minus-MFVI objective difference, which equals the reverse-KL difference. Dashes indicate unavailable absolute KL values at larger dimensions.

\begin{table}[htbp]
\centering\small\setlength{\tabcolsep}{4pt}
\caption{One-group mean errors (MFVI $\to$ mixture) and objective differences.}
\label{tab:formal-posterior}
\begin{tabular*}{\linewidth}{@{\extracolsep{\fill}}rrccc@{}}
\toprule
$p$ & $\rho$ & Reverse KL & $\Delta_{\mathrm{MF}}$ & PIP error\\
\midrule
10 & 0.7 & $0.3773\to0.0904$ & $ -0.2869 $ & $0.0230\to0.0068$\\
10 & 0.9 & $0.6916\to0.3378$ & $ -0.3538 $ & $0.0634\to0.0365$\\
20 & 0.7 & -- & $-0.1679$ & $0.01080\to0.00638$\\
20 & 0.9 & -- & $-0.2015$ & $0.01923\to0.01330$\\
30 & 0.7 & -- & $-0.1447$ & $0.00741\to0.00521$\\
30 & 0.9 & -- & $-0.2111$ & $0.01111\to0.00836$\\
100 & 0.7 & -- & $ -0.1352 $ & $0.00282\to0.00276$\\
100 & 0.9 & -- & $ -0.1160 $ & $0.00290\to0.00297$\\
\bottomrule
\end{tabular*}
\end{table}

\begin{table}[htbp]
\centering\footnotesize\setlength{\tabcolsep}{2pt}
\caption{One-group mixture-minus-MFVI differences.}
\label{tab:formal-paired}
\begin{tabular*}{\linewidth}{@{\extracolsep{\fill}}rrrcccc@{}}
\toprule
$p$ & $\rho$ & $R$ & Objective & PIP error & Support TV & Prediction MSE\\
\midrule
10 & 0.7 & 50 & \shortstack[c]{$-0.28687$\\$[-0.33610,-0.23764]$} & \shortstack[c]{$-0.01623$\\$[-0.01906,-0.01340]$} & \shortstack[c]{$-0.12059$\\$[-0.13957,-0.10161]$} & \shortstack[c]{$0.00182$\\$[-0.00255,0.00620]$}\\[2pt]
10 & 0.9 & 50 & \shortstack[c]{$-0.35385$\\$[-0.39479,-0.31291]$} & \shortstack[c]{$-0.02689$\\$[-0.03110,-0.02268]$} & \shortstack[c]{$-0.16151$\\$[-0.18223,-0.14079]$} & \shortstack[c]{$0.00001$\\$[-0.00285,0.00288]$}\\[2pt]
20 & 0.7 & 49 & \shortstack[c]{$-0.16792$\\$[-0.19062,-0.14522]$} & \shortstack[c]{$-0.00443$\\$[-0.00529,-0.00357]$} & \shortstack[c]{$-0.04453$\\$[-0.05858,-0.03047]$} & \shortstack[c]{$0.00010$\\$[-0.00249,0.00269]$}\\[2pt]
20 & 0.9 & 37 & \shortstack[c]{$-0.20154$\\$[-0.24994,-0.15314]$} & \shortstack[c]{$-0.00593$\\$[-0.00800,-0.00387]$} & \shortstack[c]{$-0.06982$\\$[-0.09718,-0.04247]$} & \shortstack[c]{$-0.00601$\\$[-0.02018,0.00817]$}\\[2pt]
30 & 0.7 & 49 & \shortstack[c]{$-0.14470$\\$[-0.16735,-0.12204]$} & \shortstack[c]{$-0.00220$\\$[-0.00275,-0.00165]$} & \shortstack[c]{$-0.02945$\\$[-0.04168,-0.01722]$} & \shortstack[c]{$0.00344$\\$[0.00159,0.00529]$}\\[2pt]
30 & 0.9 & 29 & \shortstack[c]{$-0.21109$\\$[-0.25470,-0.16748]$} & \shortstack[c]{$-0.00275$\\$[-0.00347,-0.00202]$} & \shortstack[c]{$-0.02888$\\$[-0.04197,-0.01579]$} & \shortstack[c]{$-0.00191$\\$[-0.00729,0.00347]$}\\[2pt]
100 & 0.7 & 45 & \shortstack[c]{$-0.13524$\\$[-0.16578,-0.10469]$} & \shortstack[c]{$-0.00006$\\$[-0.00019,0.00008]$} & \shortstack[c]{$-0.00148$\\$[-0.00418,0.00123]$} & \shortstack[c]{$0.00492$\\$[-0.00293,0.01277]$}\\[2pt]
100 & 0.9 & 24 & \shortstack[c]{$-0.11602$\\$[-0.14214,-0.08990]$} & \shortstack[c]{$0.00007$\\$[-0.00010,0.00023]$} & \shortstack[c]{$-0.00061$\\$[-0.00171,0.00049]$} & \shortstack[c]{$0.00138$\\$[-0.00086,0.00361]$}\\[2pt]
\bottomrule
\end{tabular*}
\end{table}

Across $p=10,20,30$, PIP error improves on 263 of 264 reference-eligible datasets and support TV on 262. At $p=100$, neither paired interval excludes zero at either correlation. Prediction MSE increases by $0.00344$ at $p=30$, $\rho=0.7$ (95\% interval $[0.00159,0.00529]$); the other one-group prediction intervals include zero.

At $p=10$, MCMC estimates have smaller PIP, support and covariance errors than the mixtures. For example, at $\rho=0.9$, their mean errors are $0.00221$, $0.0090$ and $0.0033$, versus $0.03651$, $0.1834$ and $0.0790$ for mixtures. These compare posterior summaries; an empirical MCMC law has infinite reverse KL to the continuous-slab posterior. Improved approximation also does not uniformly improve thresholded selection: at $\rho=0.7,0.9$, MFVI and mixtures share true-positive rates $0.992,0.972$, while false-discovery proportions rise from $0.0330,0.0393$ to $0.0497,0.0493$. Active-coefficient 95\% coverage is $0.956$ for both at $\rho=0.7$ and rises from $0.920$ to $0.928$ at $\rho=0.9$.

\subsection{Controlled objective and refinement comparisons}
\label{app:controlled-local}
We reuse replicates $0$--$19$ at each correlation in the two-group design, with exact references, giving 60 datasets and 480 fits across the two caps. The subset, initialization, numerical controls and primary KL endpoint were fixed before the cap-16 study; the cap-128 follow-up was specified before its fits. Both studies condition on a candidate prepared using the direct objective and do not reproduce another method's complete adaptive pipeline.

\paragraph{Common candidate and fitting arms.}
Take the first four components in the saved stagewise fit's append order and renormalize their weights. If fewer are available, duplicate a largest-weight component and split its weight equally, breaking ties by stored order. Append the residual-guided product from Appendix~\ref{app:reference-adaptive} with weight $0.1$, multiplying incumbent weights by $0.9$. All arms load identical arrays. Thirteen datasets require padding; nine have four identical incumbent shapes, making the frozen and stagewise feasible families coincide. These cases remain in every summary.

Direct joint frees all component parameters and weights. Augmented joint frees the same coordinates but evaluates the full Bernoulli/Gaussian mixture on $(\gamma,\beta^+)$ from Section~\ref{sec:direct-augmentation}. Direct frozen fixes the first four shapes but frees all weights and the fifth shape; direct stagewise also fixes the first four relative weights. No arm changes the nominal $K=5$.

\paragraph{Numerical controls and assessment.}
Each run allows 60 seconds of local refinement, at most 16 or 128 refreshes, and 25 L-BFGS-B iterations per refresh, stopping after four consecutive unsuccessful refreshes. Training starts with $2^{10}$ points per component and increases after failures as in Appendix~\ref{app:reference-adaptive}. Local validation uses three scrambles of $2^{12}$ points, tolerance $10^{-4}$ and three-standard-error acceptance. All arms bound full augmented-component and weight-vector KL by $0.25$, require relative effective sample size (ESS) at least $0.5$ and mass within $0.05$ of one for both direct and augmented importance ratios, and share step scales and parameter bounds.

Acceptance and final validation use each arm's own objective. Four separate scrambles of $2^{14}$ points per component compare the fit with its initial candidate, which is returned if improvement is unresolved. Four fresh assessment scrambles evaluate direct posterior KL, shared across arms and caps for paired precision but disjoint from fitting and validation. Exact references and truth never select fitting steps. The cap-128 runs restart from the same candidates and seeds; their first 16 refreshes match the earlier histories within $10^{-10}$ after excluding timing fields.

\paragraph{Primary outcomes and budget sensitivity.}
Table~\ref{tab:controlled-local-means} gives absolute KL means; the primary paired contrasts and intervals are in main-text Table~\ref{tab:controlled-kl-contrasts}. At both caps, direct joint beats augmentation on all 60 datasets and each direct restriction on 51, with nine numerical ties at tolerance $10^{-10}$. Mean-contrast integration SEs are below $5.6\times10^{-5}$. Increasing the cap lowers direct joint's mean KL by $8.0\%$, $5.9\%$ and $3.0\%$ across correlations (Table~\ref{tab:controlled-budget-within}). All nine mean advantages increase; eight paired change intervals exclude zero, with augmentation at $\rho=0.7$ the exception. Augmentation improves its own objective on all datasets, but its mean direct KL increases; those three increase intervals include zero.

\begin{table}[htbp]
\centering\small\setlength{\tabcolsep}{3pt}
\caption{Mean direct posterior KL at refresh caps $16\to128$.}
\label{tab:controlled-local-means}
\begin{tabular*}{\linewidth}{@{\extracolsep{\fill}}rcccc@{}}
\toprule
$\rho$ & \textbf{Direct joint (ours)} & Augmented joint & Direct frozen & Direct stagewise\\
\midrule
0.7 & $0.324\to0.298$ & $0.917\to0.982$ & $0.368\to0.361$ & $0.372\to0.366$\\
0.9 & $0.754\to0.709$ & $1.313\to1.397$ & $0.809\to0.806$ & $0.815\to0.813$\\
0.99 & $1.866\to1.810$ & $2.393\to2.463$ & $1.903\to1.902$ & $1.908\to1.907$\\
\bottomrule
\end{tabular*}
\end{table}

\begin{table}[htbp]
\centering\footnotesize\setlength{\tabcolsep}{3pt}
\caption{Within-arm changes in direct posterior KL: cap 128 minus cap 16.}
\label{tab:controlled-budget-within}
\begin{tabular*}{\linewidth}{@{\extracolsep{\fill}}lccc@{}}
\toprule
Arm & $\rho=0.7$ & $\rho=0.9$ & $\rho=0.99$\\
\midrule
\textbf{Direct joint (ours)} & \shortstack[c]{$-0.02592$\\$[-0.03944,-0.01241]$} & \shortstack[c]{$-0.04472$\\$[-0.06655,-0.02290]$} & \shortstack[c]{$-0.05631$\\$[-0.10142,-0.01121]$}\\[4pt]
Augmented joint & \shortstack[c]{$0.06533$\\$[-0.03350,0.16415]$} & \shortstack[c]{$0.08427$\\$[-0.02185,0.19038]$} & \shortstack[c]{$0.07041$\\$[-0.00350,0.14433]$}\\[4pt]
Direct frozen & \shortstack[c]{$-0.00667$\\$[-0.01820,0.00486]$} & \shortstack[c]{$-0.00247$\\$[-0.00464,-0.00031]$} & \shortstack[c]{$-0.00098$\\$[-0.00175,-0.00022]$}\\[4pt]
Direct stagewise & \shortstack[c]{$-0.00599$\\$[-0.01639,0.00441]$} & \shortstack[c]{$-0.00282$\\$[-0.00449,-0.00115]$} & \shortstack[c]{$-0.00132$\\$[-0.00274,0.00010]$}\\[4pt]
\bottomrule
\end{tabular*}
\end{table}

\paragraph{Secondary outcomes and stopping.}
The direct-versus-augmented KL gap is mainly coefficient-conditional at cap 16 and mainly support KL at cap 128, using the decomposition below. At $\rho=0.99$, the cap-128 total contrast $-0.6530$ equals support contrast $-1.0118$ plus conditional contrast $0.3588$. At $\rho=0.9,0.99$, augmentation has smaller conditional contributions despite larger total KL. Contributions weight each fit's own support law, so this does not compare conditional errors under identical weights.

At $\rho=0.99$, the direct-joint-minus-augmented support-TV contrast changes from $0.0202$ with interval $[0.0071,0.0334]$ at cap 16 to $-0.2246$ with interval $[-0.2487,-0.2005]$ at cap 128. At the larger cap, PIP, covariance, posterior-mean and full-support-TV errors also favor direct joint over augmentation, while prediction intervals include zero. Eleven of the 90 endpoint/comparator/correlation mean contrasts change sign between caps. The accompanying records retain all endpoints and intervals at both caps.

All 480 fits complete without failure or retry, and none reaches the local time limit. Cap termination falls from 179 of 240 fits at cap 16 to 26 at cap 128; 25 of the latter are augmented fits, and 15 accept their last refresh. At each cap, direct joint, frozen and stagewise return their initial candidates in 9, 9 and 12 cases, respectively; augmentation has no such returns. Final validation and assessment are outside the local clock. These stopping outcomes do not establish global convergence.

\subsection{Complete fitting procedures}
\label{app:fitting-comparison}
\label{app:two-groups}
The three procedures share saved MFVI baselines, a ten-component cap, 60-second search and 20-second proposal allowances, and independent evaluation. Frozen refinement uses the joint solver's proposal order but fixes existing shapes after initialization, freeing all weights and the new component. A split can displace its source before freezing. Stagewise fitting holds the incumbent fixed, including relative weights, and optimizes the new product and its weight \citep{miller2017boosting}. Its first two proposals use independently seeded importance/weighted-EM initializers; if both fail and time remains, it uses the released-tilt proposal (Appendix~\ref{app:stagewise-initialization}).

\begin{table}[htbp]
\centering\small\setlength{\tabcolsep}{3.5pt}
\caption{Mean errors for the two-group design.}
\label{tab:two-group-errors}
\begin{tabular*}{\linewidth}{@{\extracolsep{\fill}}rlrrrrr@{}}
\toprule
$\rho$ & Method & Reverse KL & Support TV & PIP error & Covariance & MSE\\
\midrule
0.7 & MFVI & 0.5982 & 0.2207 & 0.0416 & 0.0495 & 0.0928 \\
0.7 & \textbf{Joint refinement (ours)} & 0.1942 & 0.0837 & 0.0170 & 0.0188 & 0.0944 \\
0.7 & Frozen refinement & 0.3124 & 0.1396 & 0.0273 & 0.0282 & 0.0939 \\
0.7 & Stagewise & 0.3751 & 0.1472 & 0.0296 & 0.0271 & 0.0934 \\
\midrule
0.9 & MFVI & 1.1828 & 0.4505 & 0.1104 & 0.1757 & 0.1066 \\
0.9 & \textbf{Joint refinement (ours)} & 0.6762 & 0.3150 & 0.0730 & 0.1282 & 0.1045 \\
0.9 & Frozen refinement & 0.9419 & 0.3735 & 0.0838 & 0.1390 & 0.1035 \\
0.9 & Stagewise & 0.7752 & 0.3056 & 0.0677 & 0.1051 & 0.1008 \\
\midrule
0.99 & MFVI & 2.4889 & 0.7396 & 0.2208 & 0.6223 & 0.0857 \\
0.99 & \textbf{Joint refinement (ours)} & 1.9687 & 0.6460 & 0.1768 & 0.5509 & 0.0848 \\
0.99 & Frozen refinement & 2.0893 & 0.6542 & 0.1558 & 0.5549 & 0.0840 \\
0.99 & Stagewise & 1.7204 & 0.5788 & 0.1324 & 0.4995 & 0.0849 \\
\bottomrule
\end{tabular*}
\end{table}

Table~\ref{tab:two-group-errors} uses 50 datasets per correlation; PIP error covers all ten predictors and support TV covers the six grouped predictors. All 450 mixture fits complete without failure or final fallback. Joint refinement improves KL, support TV, PIP and covariance over MFVI on every dataset. Against stagewise, it has lower KL in 49, 38 and 11 of 50 cases as $\rho$ increases. At $\rho=0.9$, the joint-minus-stagewise covariance contrast is $0.0231$ with interval $[0.0113,0.0349]$. At $\rho=0.99$, the PIP and covariance contrasts are $0.0444$ $[0.0333,0.0555]$ and $0.0514$ $[0.0367,0.0661]$, both favoring stagewise. Frozen refinement also has smaller PIP error than joint refinement there, despite higher KL. All prediction intervals between mixture strategies include zero.

Joint refinement reaches the component cap in 149 of 150 fits. Mean component counts for frozen refinement are $9.60,9.08,7.44$ and for stagewise $5.32,6.36,7.40$; common search limits therefore do not imply equal sizes. In the supplementary one-group ablation (20 datasets per correlation), joint refinement improves KL over frozen fitting in all 40 cases, but its covariance error exceeds stagewise by $0.02249$ $[0.00168,0.04331]$ at $\rho=0.9$; prediction intervals include zero.

\paragraph{Support and coefficient error decomposition.}
\label{app:two-group-decomposition}
The KL chain rule gives
\begin{equation}
\KL(Q\|\Pi)=\KL(Q_\gamma\|\Pi_\gamma)+\sum_s Q_\gamma(s)\KL\{Q(\beta\mid\gamma=s)\|\Pi(\beta\mid\gamma=s)\}.
\label{eq:empirical-kl-decomposition}
\end{equation}
Support KL is computed exactly from saved probabilities; the conditional contribution is total KL minus support KL and inherits total-KL integration uncertainty. Figure~\ref{fig:two-group-decomposition} averages over 50 datasets per correlation, with shared horizontal scales. At $\rho=0.9$, the joint-minus-stagewise conditional contrast is $-0.14861$ (paired interval $[-0.16787,-0.12934]$), while the support contrast $0.04962$ has interval $[-0.00626,0.10550]$. At $\rho=0.99$, stagewise improves both contributions, with support KL accounting for approximately $85\%$ of its mean advantage. This locates residual error without attributing it to proposals, optimization or capacity.

\begin{figure}[htbp]
\centering
\includegraphics[width=1.0\linewidth]{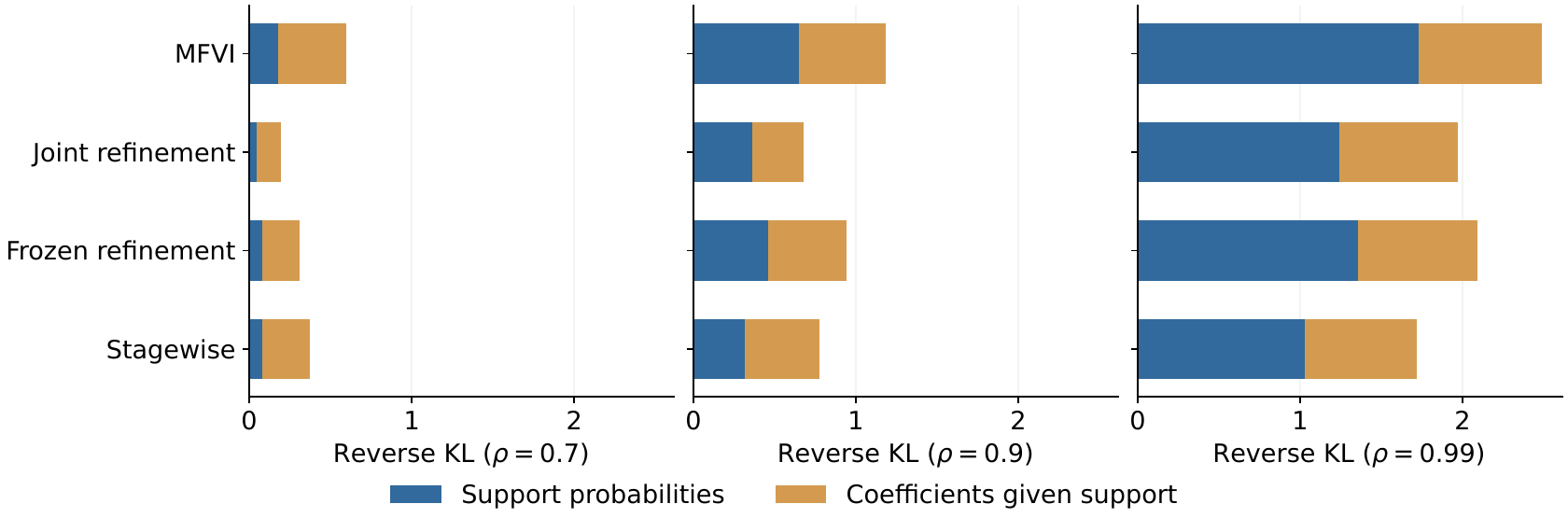}
\caption{Support and coefficient-conditional contributions to mean reverse KL.}
\label{fig:two-group-decomposition}
\end{figure}

\subsection{Reference reliability}
\label{app:reference-audit}
MCMC uses collapsed random-scan Gibbs support updates and reversible swaps, followed by conditional Gaussian coefficient draws. Four chains per dataset use 1,000 warmup sweeps and 2,000 retained draws under a 240-second cap; all 400 one-group references reach that count. Nonconstant coefficients, indicators, model size and log support mass require rank-normalized split/folded $\widehat R\leq1.01$ and bulk and tail ESS at least 400 \citep{vehtari2021rank}.

The eight grouped inclusion patterns additionally require $\widehat R\leq1.01$, bulk and raw binary ESS at least 400, and probability Monte Carlo standard error (MCSE) at most $0.01$. Constant events are flagged. This probability-specific screen was chosen after inspecting saved chains because binary tail quantiles can coincide, making quantile-based tail ESS unavailable. Combined eligibility counts at $\rho=0.7,0.9$ are $50,50$ for $p=10$, $49,37$ for $p=20$, $49,29$ for $p=30$, and $45,24$ for $p=100$.

The alternative quantile-based screen admits all 300 references at $p=10,20,30$ and $15,28$ at $p=100$. It retains PIP improvements at $p=20,30$, while neither screen resolves a mean PIP improvement at $p=100$. Four longer chains with 2,000 warmup sweeps and 8,000 retained draws on six large-case datasets all pass the probability screen; mean absolute PIP changes range from $0.00091$ to $0.00125$. Full event diagnostics and sensitivity results accompany the numerical records.

\subsection{Implementation and supplementary records}
\label{app:optimization-sensitivity}
All 100 primary $p=100$ mixture fits reach the time limit. A check on ten datasets per correlation raises the search allowance from 120 to 300 and 600 seconds under otherwise matched settings. The first increase lowers the mean objective by $0.01661$ $[0.01172,0.02151]$ and $0.01918$ $[0.00059,0.03777]$; the final increase gives only $0.00052$ and $0.00153$ more. PIP intervals mostly include zero, and all prediction intervals include zero. Complete contrasts and stopping records are retained separately.

One-group $p=10$ approximation fits use an Apple M3 MacBook Air with 8 GB; other fits use an AMD Ryzen 9 7845HX computer with 16 GB. Workers use one numerical-library thread, Python 3.12.14, NumPy 2.3.5 and SciPy 1.18.1. Wall-clock limits depend on hardware load; no speed advantage is established. Appendix~\ref{app:reference-adaptive} specifies the solver.

The code and results repository provides implementations, dataset seeds and per-dataset outcomes for reproducing the reported summaries, including reference exclusions and stopping outcomes. Its budget-sensitivity records contain both caps, all within-arm changes and all changes in contrasts. Extended experimental accounts, raw reference draws and full optimization traces are archived separately.

\FloatBarrier
\section{Algorithm details}
\label{app:reference-adaptive}

This appendix specifies Algorithm~\ref{alg:adaptive-mixture}. Adaptive fits in both designs use these settings, with fixed $\sigma^2$, $\tau^2$ and $\omega$. The controlled local variants are specified in Appendix~\ref{app:controlled-local}.

\paragraph{Mean-field initialization.}
Set $h_j=\|X_j\|^2/\sigma^2$ and $v_j=(h_j+\tau^{-2})^{-1}$. Run cyclic coordinate-ascent variational inference (CAVI) twice, with initial inclusion probabilities $\alpha_j=\omega$ and initial active means respectively zero and $v_jX_j^\top y/\sigma^2$. At each coordinate, write $m_j=\alpha_j\mu_j$ and maintain $\widetilde r=y-Xm$. In increasing predictor-index order, the updates are as follows, with superscript $\mathrm{new}$ denoting the value after the current coordinate update:
\begin{align}
 \mu_j^{\mathrm{new}}&=v_j\{X_j^\top \widetilde r/\sigma^2+h_jm_j\},\notag\\
 \alpha_j^{\mathrm{new}}&=\operatorname{logit}^{-1}\!\left\{\operatorname{logit}(\omega)+\tfrac12\log(v_j/\tau^2)+\frac{(\mu_j^{\mathrm{new}})^2}{2v_j}\right\},\notag\\
 \widetilde r^{\mathrm{new}}&=\widetilde r-X_j(\alpha_j^{\mathrm{new}}\mu_j^{\mathrm{new}}-m_j).
 \label{eq:reference-mf-updates}
\end{align}
Clip updated inclusion probabilities to $[10^{-10},1-10^{-10}]$ and keep the variances at $v_j$. Stop when the maximum change in $(\alpha,m)$ over a sweep is below $10^{-8}$ or after 1,500 sweeps. In addition to these two fits, run eight more after independently permuting the coordinate order, alternating the two initializations. Map each result back to the original predictor order, recompute its analytic objective, and retain the lowest-objective product among all ten fits as $Q_{\mathrm{MF}}$. Set the outer incumbent to this product.

\paragraph{Split proposals.}
At the outer stage with incumbent $Q^{(k)}$, select a largest-weight component $k_*$, breaking ties by its stored order. Let $J$ contain its $\min(20,p)$ largest inclusion probabilities, with predictor-index ties. Form $A_J=X_J^\top X_J/\sigma^2+\tau^{-2}I_{|J|}$ and $E_J=A_J^{-1}-\operatorname{diag}(v_{k_*j}:j\in J)$. Let $e_{\max}$ be the largest eigenvalue of $E_J$ and $b$ a corresponding unit eigenvector. With $\overline v_J=|J|^{-1}\sum_{j\in J}v_{k_*j}$, define a displacement supported on $J$ by $d_J=a\sqrt{\max\{e_{\max},0.05\overline v_J\}}\,b$, using scales $a=0.5$ and $a=1$ in successive attempts. Replace component $k_*$ by two components with means $\mu_{k_*}-d$ and $\mu_{k_*}+d$, copying its inclusion probabilities and variances and assigning each half its weight. All other components retain their parameters. The eigenvector sign merely interchanges the split components. This initializes a candidate $\widetilde Q^{(k+1)}$ with $k+1$ components, with every coordinate available during joint refinement.

\paragraph{Alternative proposal.}
If neither split is accepted and time remains, use the largest-weight component to form $\widetilde r=y-X(\alpha_{k_*}\odot\mu_{k_*})$. Choose the coordinate maximizing $(1-\alpha_{k_*j})|X_j^\top\widetilde r|/\sqrt{h_j}$, breaking ties by index. A numerical floor of $10^{-12}$ is used for $h_j$. Let $t_j=2\operatorname{sign}(X_j^\top\widetilde r)\sqrt{h_j+\tau^{-2}}$, using a positive sign for a zero inner product, and set other entries of $t$ to zero. Starting from the univariate-ridge initialization, run up to 100 CAVI sweeps with $t_j$ added inside the braces of the mean update in \eqref{eq:reference-mf-updates}. Remove the tilt and continue ordinary CAVI from that solution for up to 300 sweeps, using the same $10^{-8}$ stopping tolerance. Append the resulting product with weight $0.1$ and multiply incumbent weights by $0.9$. The temporary tilt only constructs an initializer. The subsequent mixture objective and posterior target are unchanged.

\paragraph{Local refinement.}
At each refinement refresh, use the current trial mixture as reference and draw $N$ scrambled Sobol points per component. The objective is $\widehat{\mathcal L}(\theta)=\sum_kw_k\mathcal L_k-\widehat{\mathcal J}(\theta)$, where \eqref{eq:mixture-importance} uses the full reference mixture denominator. Fix the reference during each local optimization, use unnormalized importance ratios, and differentiate as in \eqref{eq:finite-information-gradient}. Density evaluations use logarithms and log-sum-exp.

Write optimizer coordinates as $\vartheta=(u,\eta,\mu,\log v)$ with softmax weights, logit inclusion probabilities, and log variances. At reference $\vartheta^{(t)}$, optimize $\vartheta=\vartheta^{(t)}+D x$, where $D$ has unit entries except for mean entries $\sqrt{v_{kj}^{(t)}}$. Box half-widths for $x$ are respectively $1,0.5,0.5,0.4$. Intersect inclusion-logit bounds with $[-23,23]$. Start at $x=0$ and use L-BFGS-B with at most 25 iterations, relative objective tolerance $10^{-10}$, projected-gradient tolerance $10^{-5}$, and at most 15 line-search steps. A finite proposal can be checked even if the local iteration cap binds.

These fixed bounds are numerical settings for the reported experiments. If every inclusion probability lies in $[\epsilon,1-\epsilon]$ for fixed $0<\epsilon<1/2$, each component and hence the mixture satisfy $Q(S=S_0)\leq(1-\epsilon)^p$. Thus fixed clipping precludes full support-mass concentration, although this bound does not rule out consistent thresholded-PIP decisions. An asymptotic implementation intended to concentrate on $S_0$ must relax clipping at least so that $p_n\epsilon_n\to0$, together with the other approximation and optimization conditions. The theoretical families allow boundary inclusion probabilities.

\paragraph{Reference updates.}
Let $f$ count consecutive unsuccessful refreshes. Training uses $N=1024\,2^{\min(2,\lfloor f/2\rfloor)}$ points per component. Refinement stops after four unsuccessful refreshes, so the attained training sizes are 1,024 and 2,048. Each proposal allows at most 16 refreshes and a nominal 20 seconds, subject to the remaining overall budget. For a proposed displacement, start with its full length and reduce it by halves until the largest component KL to its reference and the weight-vector KL are at most $0.25$, each relative importance effective sample size is at least $0.5$, and each importance mass differs from one by at most $0.05$.

The overlap checks use source-specific ratios $a_{hb}=q_h(z_{hb})/r_h(z_{hb})$ even though objective integration uses the full mixture denominator. Their relative effective sample sizes and estimated masses are $(\sum_ba_{hb})^2/(N\sum_ba_{hb}^2)$ and $N^{-1}\sum_ba_{hb}$. There are at most twelve overlap checks with successive step reductions before a validation attempt and at most three validation attempts per refresh. A failed attempt halves the current displacement before retrying. A successful refinement resets $f$ to zero and becomes the next reference. Otherwise increment $f$ and retain the current trial.

Algorithm~\ref{alg:joint-refinement} gives the local iteration used by \textsc{JointRefine}. Within this algorithm, $K$ denotes the input mixture's component count and is fixed throughout the call. The maps \textsc{Pack} and \textsc{Unpack} convert between distribution parameters $\theta$ and the optimizer coordinates $\vartheta$ defined above. The reference is $R^{(t)}=Q_{\theta^{(t)}}$. A proposed parameter vector $\widetilde\theta^{(t)}$ becomes $\theta^{(t+1)}$ only if the overlap checks and local validation succeed. Otherwise $\theta^{(t+1)}=\theta^{(t)}$. Here $j$ indexes validation attempts and $c$ indexes overlap checks. Local validation uses scrambles separate from the fitting points and reuses them across retries within the same refresh.

\begin{algorithm}[!t]
\caption{\textsc{JointRefine}: local refinement of a fixed-size mixture}
\label{alg:joint-refinement}
\begin{algorithmic}[1]
\REQUIRE Initial mixture $Q_{\theta^{(0)}}$, model inputs, proposal deadline, seed.
\ENSURE Refined mixture with the same $K$ components and a refinement record.
\STATE $t\gets0$; $f\gets0$.
\WHILE{$t<16$, $f<4$ and time remains}
 \STATE $R^{(t)}\gets Q_{\theta^{(t)}}$; $N\gets1024\,2^{\min(2,\lfloor f/2\rfloor)}$.
 \STATE Draw $N$ points per reference component and fix the estimator in \eqref{eq:mixture-importance}.
 \STATE $\vartheta^{(t)}\gets\text{\textsc{Pack}}(\theta^{(t)})$.
 \STATE Run bounded L-BFGS-B on $\widehat{\mathcal L}$ to obtain displacement $\delta^{(t)}=Dx$.
 \STATE \textbf{if} optimization is interrupted by the deadline \textbf{then break}
 \STATE $\theta^{(t+1)}\gets\theta^{(t)}$; $a\gets1$; $\mathrm{accepted}\gets\mathrm{false}$.
 \FOR{$j=1,2,3$}
  \STATE \textbf{if} the deadline is reached \textbf{then break}
  \FOR{$c=1,\ldots,12$}
   \STATE $\widetilde\theta^{(t)}\gets\text{\textsc{Unpack}}(\vartheta^{(t)}+a\delta^{(t)})$.
   \STATE \textbf{if} all overlap checks pass \textbf{then break}
   \STATE $a\gets a/2$.
  \ENDFOR
  \IF{all overlap checks pass and time remains}
   \IF{$\text{\textsc{Validate}}(Q_{\widetilde\theta^{(t)}},Q_{\theta^{(t)}};\mathrm{local})$}
    \STATE $\theta^{(t+1)}\gets\widetilde\theta^{(t)}$; $\mathrm{accepted}\gets\mathrm{true}$.
    \STATE \textbf{break}
   \ENDIF
  \ENDIF
  \STATE $a\gets a/2$.
 \ENDFOR
 \STATE \textbf{if} accepted \textbf{then} $f\gets0$ \textbf{else} $f\gets f+1$.
 \STATE Record the refresh outcome; $t\gets t+1$.
\ENDWHILE
\RETURN $Q_{\theta^{(t)}}$.
\end{algorithmic}
\end{algorithm}

\paragraph{Acceptance and fallback.}
For fresh evaluation under a candidate mixture, sum over the conditional component label analytically. With responsibilities $\rho_k(z)=w_kq_k(z)/q_\theta(z)$ and direct points $z_{hb}$ from component $Q_h$, evaluate information as
\[
 \check{\mathcal J}(\theta)=\sum_{h=1}^{K}\frac{w_h}{N}\sum_{b=1}^{N}\sum_{k=1}^{K}\rho_k(z_{hb})\log\frac{q_k(z_{hb})}{q_\theta(z_{hb})}.
\]
Local comparisons use three independently scrambled estimates with $N=4096$. Transform common Sobol points separately under the two compared distributions. For $M$ differences $\Delta_1,\ldots,\Delta_M$, compute $\overline\Delta=M^{-1}\sum_{r=1}^M\Delta_r$ and $s_\Delta=\{\sum_{r=1}^M(\Delta_r-\overline\Delta)^2/[M(M-1)]\}^{1/2}$, and apply \eqref{eq:empirical-acceptance}. For each compared law, its mean information estimate must also lie in $[0,H(w)]$ up to the larger of $10^{-8}$ and three standard errors. The three retries within a refresh reuse its validation scrambles. These repeated comparisons form an adaptive numerical acceptance rule.

After refinement, compare the trial with the pre-expansion incumbent using three fresh scrambles with $N=8192$. Accept the first candidate passing the criterion, increment $k$ and continue; if all three proposals fail, retain the incumbent and stop. After search, use four fresh scrambles with $N=16{,}384$ to compare the incumbent with the original MFVI product, returning MFVI if improvement is unresolved. Once the return value is fixed, report an independent assessment using four fresh scrambles at the same resolution. No components are merged or pruned.

\paragraph{Stagewise initialization.}
\label{app:stagewise-initialization}
The stagewise adaptation's first two proposals use independently seeded importance/weighted-EM initializers adapted from the supplementary initialization procedure of \citet{miller2017boosting}. Draw $N=4096$ joint scrambled Sobol points from the incumbent and normalize weights proportional to $e^{-\ell(z)}p(z)/q_{\mathrm{old}}(z)$, where $p=\mathrm{d}P/\mathrm{d}\nu$. Select at most eight largest-weight observations exceeding $10/N$. Form a proposal mixing the incumbent with components centered at their sampled active coefficients, with inclusion probabilities $0.05+0.9\gamma_j$. Their slab variances are the incumbent inclusion-weighted within-component variances, floored at $(h_j+\tau^{-2})^{-1}$. Assign the selected importance weights as proposal masses, rescaled if necessary to retain at least $10^{-4}$ mass on the incumbent.

Draw another 4096 points from this proposal and reweight to the same target. Initialize the new component at the largest-weight point and its mixing fraction at $0.1$. Perform 20 importance-weighted EM iterations with the incumbent distribution clamped, updating the new fraction and its Bernoulli--Gaussian sufficient statistics. Fractions are clipped to $[10^{-4},1-10^{-4}]$, inclusion probabilities to $[10^{-6},1-10^{-6}]$, and variances to $[0.1(h_j+\tau^{-2})^{-1},4\tau^2]$. These initialization heuristics are adaptation choices. If both proposals fail acceptance and time remains, append the ordinary released-tilt proposal. Subsequent restricted refinement, trust checks, expansion validation and final fallback use the same rules as the joint-refinement solver. Code tests check the reduced-coordinate derivatives and preservation of the incumbent shapes and relative weights.

\paragraph{Execution settings.}
The master numerical seed, proposal order, and deterministic offsets for stages, refreshes, validation replicates, and component indices are recorded in the source snapshot and run manifest. Training, expansion validation, final fallback, and independent final evaluation use separate seed offsets. A scrambled Sobol point has $2p$ coordinates. Threshold the first $p$ for inclusions and transform the remaining coordinates through the standard normal quantile, clipping uniforms to $[10^{-14},1-10^{-14}]$. Conditional label expectations are summed analytically. Final evaluation uses batches of at most 1,024 points to limit memory.

The nominal search clock, 60 seconds at $p=10$ and 120 seconds at $p=20,30,100$, starts after the externally computed multistart mean-field baseline and includes proposals, refinements, and expansion validation. The optimizer and search check deadlines before further work, but an ongoing evaluation or validation can finish and accept a step after its nominal deadline. Final fallback validation is performed after the search and is included in recorded fit time. The independent final evaluation, exact posterior calculations where applicable, and posterior-summary evaluation are timed separately. The worker timeout is 600 seconds for fitting and MCMC; archived pipelines with additional comparisons use 700 seconds. The experiment records retain the returned and candidate component counts, fallback decision, validation traces, stopping causes, and the completion status of every planned dataset.

\end{document}